\documentclass[11pt,letterpaper]{amsart}
\usepackage{preamble}
\counterwithin{equation}{section}
\begin{document}

\title[The Entropic Sum-Product Phenomenon]{The Entropic Sum-Product Phenomenon}

\author[Rupert Li]{Rupert Li}
\address[]{Department of Mathematics, Stanford University, Stanford, CA 94305, USA}
\email{rupertli@stanford.edu}

\begin{abstract}
Let $X,X'$ be independent and identically distributed discrete real-valued random variables of finite Shannon entropy, and write $H(X)$ for the Shannon entropy of $X$.
We prove that
\[
  \max\set{H(X+X'),\,H(XX')} \ge \frac87 H(X)-O(\log H(X)).
\]
This is the entropic analog of the celebrated sum-product phenomenon, and answers a question of Goh, which simply asked for a coefficient strictly larger than 1.
An example by the author, Gavalakis, and Kontoyiannis showed the coefficient cannot exceed $\frac43$.
Previous work by Gavalakis, Goh, and Kontoyiannis was able to prove a result of a weaker form, which could not translate to a coefficient strictly larger than 1 because of examples where the min-entropy is significantly smaller than the Shannon entropy.
By splitting the distribution of $X$ into uniform pieces, which costs $O(\log H(X))$ entropy, we obviate this issue, establishing a coefficient of $\frac{10}{9}$.
We augment this to $\frac87$ by adapting the work of Solymosi, which established the combinatorial sum-product phenomenon with coefficient $\frac43$ by bounding the multiplicative energy, to the entropy setting, again via a uniformization technique. 
\end{abstract}

\maketitle

\section{Introduction}\label{section:introduction}
Let $A\subset\R$ be a finite set of real numbers.
The sum-product phenomenon of Erd\H{o}s and Szemer\'edi \cite{ErdosSzemeredi1983} asserts that $A$ cannot be simultaneously close to an arithmetic progression and to a geometric progression, and so cannot have both its sumset $A+A=\set{a_1+a_2:a_1,a_2\in A}$ and its product set $A\cdot A=\set{a_1a_2:a_1,a_2\in A}$ of size close to $\abs A$.
Formally, in 1983 they proved there exists some constant $\eps>0$ such that
\[ \max\set{\abs{A+A},\abs{A\cdot A}}\geq\abs{A}^{1+\eps-o(1)}, \]
where $o(1)$ denotes a term that goes to 0 as $\abs{A}\to\infty$.
The original result did not provide an explicit value for $\eps$ and only holds for subsets of the integers, i.e., $A\subset\Z$, and there has been much subsequent work on proving this result for larger and larger $\eps$.
Two milestone results in this field were the work of Elekes \cite{Elekes1997} in 1997, who proved the result for $\eps=\frac14$ using the Szemer\'edi--Trotter incidence theorem \cite{Szekely1997,SzemerediTrotter1983}, and the work of Solymosi \cite{Solymosi2009} in 2009 who proved the result for $\eps=\frac13$, by specifically proving the following theorem.
\begin{theorem}[\protect{\cite[Theorem 2.1]{Solymosi2009}}]\label{thm:solymosi}
Let $A\subset\R$ be a finite set of positive reals. Then
\[
  \abs{A\cdot A}\,\abs{A+A}^2 \ge \frac{\abs A^4}{4\ceil{\log\abs A}}.
\]
\end{theorem}
In 2015, Konyagin and Shkredov \cite{KonyaginShkredov2015} slightly improved upon Solymosi's $\eps=\frac13$ to $\eps=\frac{1}{3}+\frac{1}{20598}$, and the current published record is $\frac13+\frac{2}{1167}$ by Rudnev and Stevens \cite{RudnevStevens2022}, with subsequent preprints by Bloom \cite{Bloom2025} and then Cushman \cite{Cushman2025} improving this to $\eps=\frac13+\frac{2}{951}$ and $\eps=\frac13+\frac{10}{4407}$, respectively.

The trivial upper bound on $\eps$ is 1, and Erd\H os and Szemer\'edi \cite{ErdosSzemeredi1983} in fact conjectured that the result is true for $\eps=1$, i.e., that
\[ \max\set{\abs{A+A},\abs{A\cdot A}}\geq\abs{A}^{2-o(1)}. \]
A recent breakthrough by Bloom, Sawin, Schildkraut, and Zhelezov \cite{BSSZ2026} showed this sum-product conjecture is false over $\R$; the conjecture is still open for $A\subset\Z$.

Throughout, $\log$ denotes the base-2 logarithm, and all entropies are in bits.
The natural logarithm will be denoted $\ln$.
For a discrete random variable $X$ taking values in a set $A$, Shannon's Asymptotic Equipartition Property (AEP) lets us heuristically interpret the entropy $H(X)$ as the log-cardinality of the ``essential support'' of $X$, the probabilistic analog of $\log\abs{A}$.
Based on this heuristic correspondence, Ruzsa \cite{Ruzsa2009} in 2009 and Tao \cite{Tao2010} in 2010 began exploring this connection between combinatorics and information theory in the context of sumset inequalities.
Interestingly, they found that if one took a sumset inequality, such as the classical Ruzsa triangle inequality \cite{Ruzsa1978}
\[ \abs{A-C}\leq\frac{\abs{A-B}\abs{B-C}}{\abs{B}}, \]
and replaced discrete subsets with independent discrete random variables and log-cardinalities with entropies, miraculously one would often obtain a legitimate new entropy inequality, such as the entropic Ruzsa triangle inequality \cite{Ruzsa2009}
\begin{equation}\label{eq:entropic_Ruzsa_triangle}
    H(X-Z) \leq H(X-Y) + H(Y-Z) - H(Y)
\end{equation}
for independent $X$, $Y$, and $Z$.
This correspondence is rich and nontrivial, often with neither the combinatorial nor entropy version of these bounds implying each other.
We refer interested readers to \cite{LGK} and references therein for a modern overview of this correspondence, including the differential entropy setting for continuous random variables.

Recently, we have also seen the field of combinatorics use entropy inequalities and results as the primary tool in the argument, so this dialogue has been fruitful in both directions.
Likely the most prominent example of entropy tools being used for sumset results is the recent work of Gowers, Green, Manners, and Tao \cite{GGMT2023conjecture,GGMT2024marton}, which proves a conjecture of Marton, widely known in the literature as the \emph{polynomial Freiman--Ruzsa conjecture}, first in characteristic 2 and then for arbitrary additive groups with bounded torsion.

Bringing the sum-product phenomenon into this cardinality-entropy correspondence, Goh \cite{Goh2024} asked whether an entropic sum-product phenomenon exists.
\begin{conjecture}[\protect{\cite[Conjecture 18]{Goh2024}}]\label{conjecture:main}
There is a $\delta>0$ such that, for every discrete real-valued
random variable $X$ with independent copy $X'$ and finite entropy $H(X)<\infty$,
\begin{equation*}
  \max\set{H(X+X'),\,H(XX')} \ge \paren{1+\delta-o(1)}H(X),
\end{equation*}
where $o(1)$ denotes a term that goes to 0 as $H(X)\to\infty$.
\end{conjecture}
Note that if $X$ is uniform on a finite set $A$, one has
\[ H(X)=\log\abs{A}, \quad H(X+X')\leq\log\abs{A+A}, \quad\text{and}\quad H(XX')\leq\log\abs{A\cdot A}, \]
so an entropic sum-product phenomenon would imply its combinatorial counterpart at $\eps=\delta$, and so in this case the two statements are logically comparable, with the entropic statement being the stronger one.

The author, Gavalakis, and Kontoyiannis \cite[Example 3]{LGK} showed that if such a $\delta$ exists, it must be at most $\frac{1}{3}$, by considering the example where $X=0$ with probability $\frac{1}{3}$ and otherwise is uniformly distributed in $\set{1,2,\dots,n}$.
Thus, the optimal (maximal) $\delta$ lies in $\left[0,\frac13\right)$, while the optimal $\eps$ from the combinatorial setting lies in $\left[\frac13+\frac{10}{4407},1\right)$, and the two are thus provably not equal.

The first quantitative results towards \cref{conjecture:main} are due to Gavalakis, Goh,
and Kontoyiannis \cite{GGK}, though the conjecture remained open, i.e., they could not establish any $\delta>0$.
Here the min-entropy is denoted $H_\infty(X)=-\log\max_x\P(X=x)$.
The first result is for random variables over an arbitrary field $F$.
\begin{theorem}[\protect{\cite[Theorem 1.3]{GGK}}]\label{thm:GGK_1_3}
There exist absolute constants $K_1$ and $K_2$ such that the following holds.
Let $X$ be a finitely supported random variable taking values in a field $F$,
and let $X'$ be an independent copy of $X$. If $F$ has characteristic $p>0$,
assume further that $H_\infty(X)\le\frac23\log p-K_1$. Then
\[
  \max\set{H(X+X'),\,H(XX')} \ge \frac{4H(X)+3H_\infty(X)}6-K_2 .
\]
\end{theorem}
The second result specializes to real-valued random variables, with a slightly larger ``$\delta$'' value of $\frac15$ instead of $\frac16$.
\begin{theorem}[\protect{\cite[Theorem 1.4]{GGK}}]\label{thm:GGK_1_4}
There exists an absolute constant $K_3$ such that the following holds. If $X$
is a finitely supported real-valued random variable with $\P(X=0)=0$, and $X'$
is an independent copy of $X$, then
\[
  \max\set{H(X+X'),\,H(XX')} \ge \frac{4H(X)+2H_\infty(X)}5
  -K_3\log H_\infty(X).
\]
\end{theorem}

Both bounds involve the min-entropy.
Since $H_\infty(X)$ may be $O(1)$ while $H(X)$ is arbitrarily large (a single atom of constant mass suffices, so the $\delta\leq\frac13$ example possesses this property), neither implies \cref{conjecture:main}.

Our main theorem answers \cref{conjecture:main} affirmatively with
the explicit constant $\delta=\frac17$.

\begin{theorem}[Entropic sum-product phenomenon]\label{thm:main}
Let $X,X'$ be independent and identically distributed, discrete,
real-valued random variables with finite entropy $H(X)<\infty$.
There are absolute constants $c_1,c_2$ and $H_0$ such that, if $H(X)\ge H_0$, then
\begin{equation}\label{eq:main-effective}
  \max\set{H(X+X'),\,H(XX')} \ge \frac87H(X)-c_1\log H(X)-c_2
  -\frac12\log C_R,
\end{equation}
where $C_R\ge1$ is the absolute constant of the point-plane incidence bound in \cref{def:KR}. In particular,
\begin{equation}\label{eq:main-asymptotic}
  \max\set{H(X+X'),\,H(XX')} \ge \paren{1+\frac17-o(1)}H(X).
\end{equation}
\end{theorem}
As noted in \cref{rmk:characteristic}, some of our arguments continue to hold for discrete random variables over an arbitrary field of characteristic 0, and in particular for such fields, \eqref{eq:main-asymptotic} holds with $\delta=\frac19$ instead of $\delta=\frac17$.
In \cref{rmk:characteristic} we also discuss what must be changed to address fields of positive characteristic, and we believe that our framework should allow for some result of the form \eqref{eq:main-asymptotic} to hold for such fields (likely requiring $p$ being sufficiently large), perhaps with a worse $\delta<\frac19$.

\Cref{thm:main} is proved in two stages. The body of the paper treats finitely
supported $X$, the theorem being restated and proved in that case as
\cref{thm:main-body}, where one may take $H_0=5$, $c_1=18$ and $c_2=63$; the
general discrete case immediately follows from it by the simple truncation argument of
\cref{lem:truncation}, and holds with $H_0=6$, $c_1=18$ and $c_2=65$, restated as \cref{cor:main-general}.

The main technical steps are the following two results, which we call
\emph{linear entropic Elekes--Ruzsa theorems} as explained below, and which we believe will be of independent interest.
Write $\widetilde X$ for a discrete real-valued random variable with $\P(\widetilde X=0)=0$; we call such a law \emph{zero-free}.
\begin{theorem}[Small-doubling linear entropic Elekes--Ruzsa]\label{thm:eer}
Let $\widetilde X$ be zero-free with entropy $H(\widetilde X)\in[5,\infty)$, let $\widetilde X'$ be an independent copy
of $\widetilde X$, and define
\[
  \widetilde u=\frac{H(\widetilde X+\widetilde X')}{H(\widetilde X)}-1,
  \qquad
  \widetilde m=\frac{H(\widetilde X\widetilde X')}{H(\widetilde X)} .
\]
Then
\begin{equation*}
  \widetilde m \ge 2-6\widetilde u-\vartheta(H(\widetilde X)),
  \qquad
  \vartheta(x)=\frac{18\log x + 68}{x}.
\end{equation*}
\end{theorem}
\begin{theorem}[Large-doubling linear entropic Elekes--Ruzsa]\label{thm:baseline}
With the same assumptions and notation as \cref{thm:eer},
\[ \widetilde m \ge \frac54-\frac34\widetilde u-\vartheta(H(\widetilde X)), \qquad \vartheta(x) = \frac{\frac34\log(1+x)+K_R+5}x, \]
where $K_R\geq2$ is an absolute constant defined in \cref{def:KR}.
\end{theorem}
We refer to the former result as small-doubling and the latter as large-doubling because the former yields a better bound for $\widetilde u\leq\frac17$ while the latter yields a better bound for $\widetilde u\geq\frac17$.
Note that as $H(\widetilde X)\to\infty$, either of these two results already proves that $\max\set{\widetilde u+1,\widetilde m}\geq\frac87-o(1)$, i.e., already proves the entropic sum-product phenomenon for $\delta=\frac17$ as in \cref{thm:main}, for zero-free laws.
Unfortunately, as discussed in \cref{sec:main}, reincorporating the possibility of an atom at zero has a significant detrimental effect on $\delta$, and it is in fact the \emph{combination} of these two results that miraculously yields $\delta=\frac17$ for general laws, i.e., \cref{thm:main}; for example, only using \cref{thm:baseline} will yield $\delta=\frac19$ after reincorporating zeros (see \cref{rmk:1_9}).

For finitely supported $\widetilde X$, \cref{thm:eer} is restated and proved as \cref{cor:eer-eff}, where $\vartheta(x)=\frac{18\log x+57}{x}$ for $x\ge4$, and then by the truncation argument of \cref{lem:truncation} extended to the general case.
Similarly, the finite support version of \cref{thm:baseline} is proved as \cref{prop:baseline}, and then extended to the general case by \cref{lem:truncation}.

The name refers to the following theorem of Elekes and Ruzsa \cite{ElekesRuzsa2003}, the combinatorial statement of the same shape: ``few sums force many products.''
\begin{theorem}[{\cite{ElekesRuzsa2003}}; see also \cite{Elekes1997}]
There is an absolute constant $c>0$ such that, for every finite
$A\subset\R$ with $\abs{A+A}\le K\abs A$,
\[
  \abs{A\cdot A} \ge \frac{c\,\abs A^2}{K^4\log\abs A} .
\]
\end{theorem}

\begin{remark}\label{rmk:ggk16}
\Cref{thm:eer} is the entropic analog, but crucially linear in the relative additive doubling $\widetilde u$.
The existing literature from additive combinatorics, i.e., Freiman-type inverse sumset theory which studies the structure of $A$ given $\abs{A+A}\leq K\abs{A}$, considers $K$ as a constant, which translates to an \emph{additive} constant in the entropic setting, rather than a \emph{multiplicative} constant, which we need to prove $\delta>0$.
The quantitative dependence on $K$ makes these results unsuitable for our setting.
An entropic result of this few-sums-many-products form was previously proven in \cite[Theorem 1.6]{GGK}, but it is of this Freiman-type additive constant form, assuming $H(X+X')\leq H(X)+C$ and that $H(X)$ is sufficiently large, depending on $C$.
\end{remark}

Much of our analysis was inspired by Solymosi's result, \cref{thm:solymosi}, establishing the combinatorial result for $\eps=\frac13$.
We were not able to prove the entropic analog of \cref{thm:solymosi}, which would immediately imply \cref{conjecture:main} for $\delta=\frac13$, and would be tight by \cite[Example 3]{LGK} showing $\delta\leq\frac13$.
We conjecture that this entropic analog is true, and thus that the optimal $\delta$ is $\frac13$.
\begin{conjecture}\label{conjecture:entropic_solymosi}
    Let $X,X'$ be independent and identically distributed, discrete, real-valued random variables with finite entropy $H(X)<\infty$.
    Then
    \begin{equation}\label{eq:entropic_solymosi}
        2H(X+X') + H(XX') \geq (4-o(1))H(X),
    \end{equation}
    where $o(1)$ denotes a term that goes to 0 as $H(X)\to\infty$.
\end{conjecture}
\cref{rmk:conj-es-zf} reduces \cref{conjecture:entropic_solymosi} to the case of $X$ being zero-free.
\cref{sec:open} further discusses this conjecture and records two further problems that would each improve $\delta$ for the entropic sum-product phenomenon.
\subsection{Proof summary and section overview}
Gavalakis, Goh, and Kontoyiannis \cite{GGK} apply a uniformization lemma of Vadhan \cite{Vadhan2012} to split their distribution into a mixture of uniform distributions, the granularity of which is limited by the min-entropy.
Then they use a point-plane incidence bound of Rudnev \cite{Rudnev2018} to prove lower bounds on $H(A(B+C))$, where $A$, $B$, and $C$ are these uniform distributions.
Combined with upper bounds by M\'ath\'e and O'Regan \cite{MO2025} on $H(X(Y+Z))$ in terms of the sum and product entropies, they are able to establish their sum-product results. 

By first applying a dyadic decomposition to our random variable, which partitions its support into rough level sets of the probability mass function, and then applying the uniformization lemma, we avoid this limitation of the min-entropy, at a cost of $O(\log H(X))$, which is negligible for the purposes of the sum-product phenomenon.
This causes our analogous $A$, $B$, and $C$ to potentially have very different sizes, something the original lower bound on $H(A(B+C))$ is unable to handle, so we must generalize their arguments.
The previously stated upper bounds are actually incorrect, though the main results of Gavalakis, Goh, and Kontoyiannis \cite{GGK} are salvageable (see \cref{rmk:erratum}), so we also have to establish new, corrected upper bounds.
This adaptation by itself is able to prove the case $\delta=\frac19$ via establishing \cref{thm:baseline}.

In the hopes of improving this all the way to the tight $\delta=\frac13$, which was proven in the combinatorial setting by Solymosi \cite{Solymosi2009}, we studied their combinatorial argument bounding the multiplicative energy by sumset cardinalities.
While we were unsuccessful in fully transferring the argument to the entropic setting, this did enable us to prove \cref{thm:eer}, our small-doubling linear entropic Elekes--Ruzsa theorem.
This protects against the weakness of \cref{thm:baseline} at small doublings, improving $\delta$ to $\frac17$.

The dyadic uniformization technique is used in both parts of the argument, i.e., both small-doubling and large-doubling, and to do so we need the finite support assumption.
Fixing this costs essentially nothing at the level of our main theorems, as truncating our random variables to some large but finite support has negligible impact on their entropies, and we handle this at the end of our paper. 

\cref{sec:prelim} introduces basic entropy properties and notation, states the point-plane incidence bound, and establishes our combined uniformization technique.
\cref{sec:solymosi} establishes the entropic analog of Solymosi's \cite{Solymosi2009} bound in the clean case of near-uniform distributions, which \cref{sec:eer} uses to obtain the small-doubling linear entropic Elekes--Ruzsa theorem by tracking all of the (pairs of) pieces from the uniformization technique.
\cref{sec:incidence} proves the lower bounds on $H(A(B+C))$ via the incidence argument, while \cref{sec:upper} proves the upper bounds and combines the two to yield our large-doubling linear entropic Elekes--Ruzsa theorem.
\cref{sec:main} combines these two results, reintroduces the atom at zero, and removes the finite support hypothesis to prove \cref{thm:main}.
\cref{sec:open} discusses \cref{conjecture:entropic_solymosi} and poses some avenues for further research towards improving $\delta=\frac17$.

\section{Preliminaries}\label{sec:prelim}
\subsection{Entropy}
We refer readers to \cite{CoverThomas} for an overview of entropy and information theory, but provide a brief self-contained introduction here.
Throughout the paper, all random variables are discrete.
For a discrete $X$ with mass function $p$ we
write
\[
  H(X)=-\sum_x p(x)\log p(x),\qquad
  H_\infty(X)=-\log\max_x p(x),\qquad
  H_2(X)=-\log\sum_x p(x)^2,
\]
where $H_2(X)=-\log\P(X=X')$ for $X'$ an independent copy of $X$.
It is straightforward to see and well-known that
\[ H_\infty(X)\le H_2(X)\le H(X), \]
with equality throughout when $X$ is uniform on a finite set.
We write
\[ h(\theta)=-\theta\log\theta-(1-\theta)\log(1-\theta) \]
for the binary entropy function, with $0\log 0:=0$, so $0\le h\le1$ for $\theta\in[0,1]$.
For finite nonempty $S$, let $U_S$ denote the uniform random variable on $S$, and for $\P(E)>0$ we write $H(X\mid E)$ for the entropy of $X$ under $\P(\cdot\mid E)$.
For jointly distributed discrete random variables $X$ and $Y$, the conditional entropy $H(X|Y)$ is given by
\[ H(X|Y) = \sum_y \P(Y=y)H(X|Y=y). \]

Henceforth let $X$ be real-valued, $X'$ be an independent copy, and for brevity of notation let $H=H(X)$.
When $\P(X=0)<1$ we write $\widetilde X$ for the zero-free random variable whose law is that of $X$ conditioned on $\{X\ne0\}$.
Similarly denote $\widetilde H=H(\widetilde X)$.
The four ratios that organize the whole argument are
\[
  u=\frac{H(X+X')}{H}-1,\qquad m=\frac{H(XX')}{H},\qquad
  \widetilde u=\frac{H(\widetilde X+\widetilde X')}{\widetilde H}-1,\qquad
  \widetilde m=\frac{H(\widetilde X\widetilde X')}{\widetilde H},
\]
where $X'$ is an independent copy of $X$, and $\widetilde X'$ is an independent copy of $\widetilde X$, and these ratios are defined whenever the entropy in the denominator is positive.
For a finite subset $P$ of an abelian group we write
\[ g(P)=2\log\abs{P}-\log\abs{P+P} \]
for its \emph{doubling defect}, which in some sense quantifies the number of additive collisions in $P$, i.e., is large when $\abs{P+P}$ is small.
 
We collect the following standard entropy facts.
\begin{lemma}[Elementary entropy facts]\label{lem:entropy-facts}
Let $W,T,R,S$ be discrete random variables with finite entropies.
\begin{enumerate}[leftmargin=3.2em]
    \item[\textup{(E1)}] \emph{(Conditioning reduces entropy.)}
    \[ H(W)\ge H(W\mid T)=\sum_t\P(T=t)H(W\mid T=t), \]
    so the mutual information
    \[ I(W;T) := H(W)-H(W\mid T) \geq 0. \]
    \item[\textup{(E2)}] \emph{(Chain rule.)}
    \[ H(W,T)=H(T)+H(W\mid T). \]
    In particular, if $T=f(W)$ for some function $f$ then
    \[ H(W)=H(T)+H(W\mid T)=H(T)+\sum_t\P(T=t)H(W\mid T=t). \]
    \item[\textup{(E3)}] \emph{(Data processing and subadditivity.)}
    \[ I(f(R);g(S))\le I(R;S) \]
    for all functions $f,g$; consequently $H(f(W))\le H(W)$ and
    \[ H(W,T)=H(W)+H(T)-I(W;T)\le H(W)+H(T). \]
    \item[\textup{(E4)}] \emph{(Support bound.)} $H(W)\le\log\abs{\operatorname{supp}W}$, and
    $H(W)\ge H_2(W)\ge H_\infty(W)$.
    \item[\textup{(E5)}] \emph{(Monotonicity under independent sums.)} If $W,T$ are independent
    and take values in an abelian group, then
    \[
    H(W+T) \ge \max\{H(W),H(T)\};
    \]
    and if $W,W'$ are i.i.d., then
    \[ H(WW')\ge\P(W'\ne0)\,H(W). \]
    \item[\textup{(E5$'$)}] $H(X+X')\le H(X,X')=2H$, so that $u\in[0,1]$ whenever $H>0$;
    likewise $\widetilde u\in[0,1]$ whenever $\widetilde H>0$.
    \item[\textup{(E6)}] \emph{(Gibbs' inequality.)} Let $\nu$ be a sub-probability measure on
    the range of $W$ (that is, a measure with total measure at most 1) with $\nu(w)>0$ for all $w\in\operatorname{supp}W$. Then
    \[
    H(W) \le \E\bracket{\log\frac1{\nu(W)}}
     = -\sum_w\P(W=w)\log\nu(w).
    \]
    \item[\textup{(E7)}] \emph{(Maximum entropy of the geometric law.)} If $W$ takes values in
    $\Z_{\ge0}$ and $\E[W]=\mu<\infty$ then
    \[
    H(W) \le f(\mu):=(1+\mu)\log(1+\mu)-\mu\log\mu \le \log(1+\mu)+\log e,
    \]
    and $f$ is nondecreasing on $[0,\infty)$.
\end{enumerate}
\end{lemma}
 
\begin{proof}
(E1)--(E6) are standard, except for possibly the second part of (E5); see \cite{CoverThomas}.
For the second part of (E5),
\[ H(WW') \geq H(WW'|W') = \sum_{w'}\P(W'=w')H(w'W) = \P(W'\neq 0) H(W). \]
For (E7), let $\theta=\frac{\mu}{1+\mu}\in[0,1)$ and let $\nu$ be the geometric law $\nu(n)=(1-\theta)\theta^n$ on $\Z_{\ge0}$, whose mean is $\frac{1}{1-\theta}-1=\mu$.
It is positive everywhere on $\Z_{\ge0}$, so (E6) applies and gives
\[
  H(W) \le -\log(1-\theta)-\E[W]\log\theta
  =\log(1+\mu)+\mu\log\frac{1+\mu}{\mu}
  =(1+\mu)\log(1+\mu)-\mu\log\mu=f(\mu).
\]
Moreover
\[ \mu\log\paren{1+\frac{1}{\mu}}=\mu\log e\cdot\ln\paren{1+\frac1\mu}\le\log e \]
by
$\ln(1+x)\le x$, so
\[ f(\mu)\le\log(1+\mu)+\log e; \]
and
\[ f'(\mu)=\log\paren{\frac{1+\mu}{\mu}}>0, \]
so $f$ is increasing.
\end{proof}

Recall the entropic Ruzsa triangle inequality, \eqref{eq:entropic_Ruzsa_triangle}, which implies
\begin{equation}\label{eq:difference-from-sum}
  H(\widetilde X-\widetilde X') \le 2H(\widetilde X+\widetilde X')-\widetilde H
   = (1+2\widetilde u)\widetilde H .
\end{equation}

Lastly, the following result, referred to as \emph{functional submodularity}, is classical.
\begin{lemma}[Functional submodularity; see, e.g., \protect{\cite{Tao2010}}]\label{lem:submod}
Let $R,S,W,T$ be jointly distributed discrete random variables with finite entropies, and
suppose that $W=f(R)=g(S)$ almost surely and $T=\varphi(R,S)$ almost surely, for some
functions $f,g,\varphi$. Then
\[ H(W)+H(T) \le H(R)+H(S). \]
\end{lemma}

\subsection{Rudnev's incidence bound}
The proofs for the combinatorial sum-product phenomenon often rely on incidence bounds, so it comes as no surprise that we find them helpful for the entropic analog.
For an arbitrary field $F$, let $P$ be a set of points and let $Q$ be a set of planes in $F^3$; we write
\[ I(P,Q)=\#\{(x,\ell)\in P\times Q:\ x\in\ell\} \]
for the number of incidences.
The following result, originally due to Koll\'ar \cite{Kollar2015}, is the prime-field analog of Rudnev's \cite{Rudnev2018} point-plane incidence bound; the version quoted here is a slight strengthening by de Zeeuw \cite{deZeeuw2016}, to which we refer for a quick proof.
This is the same result used by Gavalakis, Goh, and Kontoyiannis \cite[Lemma 3.1]{GGK}, which is how we became aware of the result.
For the classical Euclidean antecedents see \cite{Szekely1997,SzemerediTrotter1983}; for the current state of the art, and the natural place to look for an explicit constant, see \cite{RudnevStevens2022}.
Here, the notation $X\ll Y$ means there exists some absolute constant $C>0$ such that $X \leq CY$.
\begin{theorem}[{\cite[Theorem~4.1]{deZeeuw2016}}]\label{thm:incidence}
Let $F$ be a field of characteristic $p$, let $P$ be a set of points in $F^3$ and
let $Q$ be a set of planes in $F^3$. Suppose that
\begin{enumerate}[label=(\roman*)]
    \item $\abs{P}\le\abs{Q}$;
    \item either $p=0$, or $\abs{P}\ll p^2$; and
    \item no line of $F^3$ contains $k$ points of $P$ and is contained in $\ell$ planes of $Q$.
\end{enumerate}
Then
\begin{equation}\label{eq:incidence}
  I(P,Q) \ll \abs{P}^{1/2}\abs{Q}+\ell\abs{P}+k\abs{Q} .
\end{equation}
\end{theorem}
This formulation, where hypothesis (iii) is only violated by a line simultaneously meeting $P$ many times and contained in many planes of $Q$ (the latter condition not being present in Rudnev's original result), will be useful for us in \cref{sec:incidence}.
 
\begin{definition}\label{def:KR}
We let $C_R\ge1$ denote the absolute constant implied by the symbol $\ll$ in
\eqref{eq:incidence}; that is, under the hypotheses of \cref{thm:incidence},
\[
  I(P,Q) \le C_R\paren{\abs{P}^{1/2}\abs{Q}+\ell\abs{P}+k\abs{Q}}.
\]
Throughout this paper $F=\R$, so that hypothesis (ii) of \cref{thm:incidence} is vacuous. We
set
\[
  K_R = 2+\log C_R .
\]
\end{definition}
No explicit admissible value of $C_R$ appears in the literature, so $K_R$ is the only
unpinned absolute constant in this paper.

\subsection{A uniformization lemma}
The following result shows that a law with min-entropy at least $m$ is an exact convex combination of uniform laws on sets of size $2^m$.
This was also used by Gavalakis, Goh, and Kontoyiannis \cite[Lemma 3.4]{GGK}, which is how we became aware of this result.
\begin{lemma}[{\cite[Lemma~6.10]{Vadhan2012}}]\label{lem:vadhan}
Let $W$ be a random variable with finite support $A$ and $H_\infty(W)\ge m$, where $M=2^m$ is a positive integer, though $m$ need not be.
Then there exist $s\in\N$, reals $p_1,\dots,p_s\ge0$ with $\sum_{i=1}^sp_i=1$, and sets $A_1,\dots,A_s\subseteq A$ with $\abs{A_i}=M$ for every $i$, such that
\[
  \P(W=a) = \sum_{i=1}^sp_i\,\P(U_{A_i}=a)\qquad\text{for all }a\in A .
\]
\end{lemma}
The following lemma, while elementary albeit somewhat nonstandard, is crucial in enabling us to use Shannon entropy in place of min-entropy in the work of Gavalakis, Goh, and Kontoyiannis \cite{GGK}.
It cuts an arbitrary finitely supported random variable $X$ into flat, i.e., roughly uniform, pieces via dyadic intervals at a total cost of $O(\log H(X))$ bits.
It bears similarity to (prefix-free) source coding (see \cite{CoverThomas} for an overview), e.g., $J$ is roughly analogous to the optimal code length.
The idea behind this lemma is certainly not new, and we do not claim originality of it; see, e.g., \cite{CTLi2025}.
\begin{lemma}\label{lem:layers}
Let $X$ be finitely supported with mass function $p$ and entropy $H(X)>0$.
For $j\ge0$ define
\[
  A_j=\{x: 2^{-j-1}<p(x)\le2^{-j}\},
\]
let $J$ be the random index such that $X\in A_J$, and let $X_j=\operatorname{Law}(X\mid J=j)$ with mass function $p_j$, for every $j$ with $\P(J=j)>0$.
Then:
\begin{enumerate}[label=(\alph*)]
    \item the sets $A_j$ partition $\operatorname{supp}X$ and $\abs{A_j}<2^{j+1}$; moreover, for every $j$ with $\P(J=j)>0$ one has
    $\max p_j<2\min p_j$ on $A_j$, hence
    \[
    \frac1{2\abs{A_j}} \le p_j(x) \le \frac2{\abs{A_j}}\quad\text{for all }x\in A_j,
    \qquad\text{and}\qquad H_\infty(X\mid J=j)>\log\abs{A_j}-1;
    \]
    \item $\E[J]\le H(X)$ and $H(J)\le\log(1+H(X))+\log e$;
    \item
    \[ \displaystyle H(X)-H(J) \le \sum_j \P(J=j)\log\abs{A_j} \le H(X)+1; \]
    \item if $\emptyset\ne D\subseteq A_j$ then the conditional mass function
    $p_D:=p_j/p_j(D)$ on $D$ satisfies
    \[ \max p_D<2\min p_D\le2/\abs{D}; \]
    \item \emph{(sign refinement)} assume in addition that $\P(X=0)=0$, and refine
    each $A_j$ according to the sign of its elements, so that
    \[ A_{(j,\sigma)}:=\{x\in A_j:\operatorname{sgn}(x)=\sigma\} \]
    for $\sigma\in\{+,-\}$.
    Write $\{A_i\}_{i\in\mathcal I}$, $\mathcal I\subseteq\Z_{\ge0}\times\{+,-\}$, for the resulting
    partition into sign-pure sets, and let $I$ be the index with $X\in A_I$, with $X_i=\operatorname{Law}(X\mid I=i)$.
    Then (a), (c) and (d) hold verbatim for $\{A_i\}_{i\in\mathcal I}$ in place of $\{A_j\}_{j\ge0}$, and
    \[
    H(I) \le \log(1+H(X))+\log e+1 .
    \]
\end{enumerate}
\end{lemma}
\begin{proof}
(a) Clearly the $A_j$ partition $\operatorname{supp} X$, and as $p(x)>2^{-j-1}$ for any $x\in A_j$, we immediately have $\abs{A_j}<2^{j+1}$.
Fix $j$ with $\P(J=j)>0$.
The definition of $A_j$ ensures $\max p_j < 2\min p_j$ on $A_j$.
As $\sum_{x\in A_j}p_j(x)=1$ gives
$\min p_j\le\frac{1}{\abs{A_j}}\le\max p_j$, we have
\[
  p_j(x) \le \max p_j < 2\min p_j \le \frac2{\abs{A_j}},
  \qquad
  p_j(x) \ge \min p_j > \frac12\max p_j \ge \frac1{2\abs{A_j}}.
\]
The first chain gives
\[ H_\infty(X\mid J=j)=-\log\max p_j>-\log\frac{2}{\abs{A_j}}=\log\abs{A_j}-1. \]
 
(b) Pointwise,
\[ J(x)=\left\lfloor\log\frac{1}{p(x)}\right\rfloor\le-\log p(x), \] 
so $\E[J]\le H(X)$; as $J$ takes
values in $\Z_{\ge0}$, (E7) from \cref{lem:entropy-facts} gives
\[ H(J)\le f(\E[J])\le f(H(X))\le\log(1+H(X))+\log e. \]
 
(c) As $J$ is a function of $X$, by the chain rule and the support bound,
\[ H(X) = H(X,J) = H(J) + H(X|J) \leq H(J) + \sum_j \P(J=j)\log\abs{A_j}; \]
for the upper bound, $\abs{A_j}<2^{j+1}$ gives
\[ \sum_j \P(J=j)\log\abs{A_j} \le \E[J]+1\le H(X)+1 \]
by (b).

(d) This follows from the same argument as in (a).
 
(e) The proofs of (a) and (d) apply verbatim.
For (c), the proof of the lower bound applies verbatim, and for the upper bound, as $\log\abs{A_i}\leq j(i)+1$ where $j(i)$ is the dyadic index such that $A_i\subseteq A_{j(i)}$, analogous to before we have
\[ \sum_i \P(I=i)\log\abs{A_i} \leq 1 + \sum_i \P(I=i) j(i) = \E[J]+1 \leq H(X) + 1. \]
As $I$ is a function of $(J,\operatorname{sgn}(X))$, by (b) we have
\[ H(I) \leq H(J) + H(\operatorname{sgn}(X)) \leq \log(1+H(X))+\log e + 1. \qedhere \]
\end{proof}
We concatenate \cref{lem:layers} and \cref{lem:vadhan} into the following result.
\begin{lemma}\label{lem:unif-decomp}
Let $X$ be finitely supported with $H=H(X)>0$.
Then we can augment our probability space to include a random variable $L$ taking values $\ell$ in a finite set, together with sets $G_\ell\subseteq\operatorname{supp}X$ and integers $n_\ell\ge0$, such
that:
\begin{enumerate}
\item $|G_\ell|=2^{n_\ell}$;
\item $X|L=\ell\sim\Unif(G_\ell)$;
\item $H(X|L) = \E[n_L] \ge H-\Delta(H)$, where $\Delta(H)=\log(1+H)+\log e+2$.
\end{enumerate}
\end{lemma}
When we say we can augment our probability space, we simply mean that $L$ is not necessarily $X$-measurable; $X$ still retains its original distribution.
\begin{proof}
Let $A_j$, $J$, $w_j=\P(J=j)$, $a_j=\log\abs{A_j}$ be the \emph{unrefined} dyadic layer data of
\cref{lem:layers} for $X$, and let
\[
  n_j=\begin{cases}\floor{a_j}-1,&\abs{A_j}\ge2,\\[2pt] 0,&\abs{A_j}=1.\end{cases}
\]
Then $n_j\in\Z_{\ge0}$ and $2^{n_j}$ is a positive integer with $2^{n_j}\le\abs{A_j}$. For
$\abs{A_j}\ge2$, \cref{lem:layers}(a) gives
\[ H_\infty(X\mid J=j)>a_j-1\ge\floor{a_j}-1=n_j; \]
for $\abs{A_j}=1$, we have $H_\infty(X\mid J=j)=0=n_j$.
So \cref{lem:vadhan}, applied to $\mathrm{Law}(X\mid J=j)$, which
is supported in $A_j$, with parameter $n_j$, yields weights $p_{j,i}\ge0$ summing to $1$ and sets
$G_{j,i}\subseteq A_j$ with $|G_{j,i}|=2^{n_j}$ whose $p_{j,i}$-weighted mixture of uniforms on $G_{j,i}$ is $\mathrm{Law}(X\mid J=j)$.
Let $L=(J,I)$ take value $(j,i)$ accordingly, where one can consider selecting $J=j$ with probability $w_j$, then selecting $I=i$ conditioned on $J=j$ with probability $p_{j,i}$, and then sampling $X$ uniformly from $G_{j,i}$.
Then the notation $n_L$ from the lemma statement corresponds to $n_J$, and (1) and (2) are immediate.

For (3), note that $n_j>a_j-2$ in both cases.
Hence, by (b) and (c) from \cref{lem:layers},
\[
  \E[n_L]=\sum_jw_jn_j > \E[a_J]-2 \ge H-H(J)-2
   \ge H-\log(1+H)-\log e-2 = H-\Delta(H). \qedhere
\]
\end{proof}
\section{An asymmetric Solymosi inequality}\label{sec:solymosi}
For finite sets $A,B\subset\R\setminus\set{0}$ we define their \emph{multiplicative cross-energy}
\begin{equation}\label{eq:mult-energy}
  E_\times(A,B) = \#\{(a,b,a',b')\in A\times B\times A\times B:\ ab=a'b'\}.
\end{equation}
This section records Solymosi's \cite{Solymosi2009} sector-decomposition bound
\[
  E_\times(A,B) \ll \log\min\set{\abs A,\abs B}\,\abs{A+A}\,\abs{B+B}
\]
with explicit constants and assuming $A$, rather than both $A$ and $B$, is constrained to be of a single sign, i.e., \emph{sign-pure}.
 
This energy bound is essentially known, and we claim no novelty for it.
It appears, without effective constants, as \cite[Theorem 6]{KonyaginShkredov2015}, which holds without any sign constraints on $A$.
We reproduce a proof, following Solymosi \cite{Solymosi2009}, to provide explicit constants.
 
\begin{theorem}[see \protect{\cite[Theorem 6]{KonyaginShkredov2015}}]\label{thm:asym-solymosi}
Let $A\subset\R_{>0}$ and $B\subset\R\setminus\{0\}$ be finite and nonempty, and let
\[
  L_\times(A,B)=\floor{\log\min\set{\abs A,\abs B}}+1 .
\]
Then
\begin{equation}\label{eq:asym-solymosi}
  E_\times(A,B) \le 8L_\times(A,B)\,\abs{A+A}\,\abs{B+B} .
\end{equation}
Consequently, since all relevant cardinalities are unchanged when $A$ is
replaced by $-A$, the bound \eqref{eq:asym-solymosi} holds whenever $A$ is finite, nonempty
and sign-pure and $B$ is finite, nonempty and zero-free.
\end{theorem}
\begin{proof}
We follow Solymosi \cite{Solymosi2009}, tracking the constants.
For $s\neq0$, let
\[ n_s=\abs{A\cap s^{-1}B}=\#\set{(a,b)\in A\times B:\ b=sa} \]
be the number of points of the grid $A\times B$ on the ray $\{y=sx,\ x>0\}$.
Since $0\notin B$ and every element of $A$ is strictly positive, each $(a,b)\in A\times B$ lies on exactly one such ray, that of slope $s=\frac{b}{a}\ne0$, where $s$ may be negative.
Consequently $\sum_sn_s=\abs A\abs B$ and
\[
  \sum_sn_s^2=\#\set{(a,b,a',b')\in(A\times B)^2:\ \frac ba=\frac{b'}{a'}}
  =E_\times(A,B).
\]
Let $N_0=\min\set{\abs A,\abs B}$, so that $n_s\le N_0$ for every $s$ and $L_\times(A,B)=\floor{\log N_0}+1$.
Define the dyadic classes
\[
  A_j=\{s:\ 2^j\le n_s<2^{j+1}\},\qquad
  j=0,1,\dots,\floor{\log N_0},
\]
which partition the $s$ with $n_s\ge1$, so as there are $L_\times(A,B)$ possible values for $j$, there exists $j$ with
\[ \sum_{s\in A_j}n_s^2\ge \frac{E_\times(A,B)}{L_\times(A,B)}. \]
Fix such a $j$, enumerate $A_j=\{s_1<s_2<\dots<s_k\}$, and set
\[ \Pi_\ell=(A\times B)\cap\{(x,y):\ y=s_\ell x\},\qquad1\le\ell\le k, \]
so that $2^j\le\abs{\Pi_\ell}=n_{s_\ell}<2^{j+1}$ for every $\ell$. Then
\begin{equation*}
  E_\times(A,B) \le L_\times(A,B)\sum_{s\in A_j}n_s^2 < L_\times(A,B)\,k\,2^{2j+2}
   = 4L_\times\,k\,2^{2j} ,
\end{equation*}
so it suffices to prove
\begin{equation}\label{eq:solymosi-target}
  k\,2^{2j} \le 2\abs{A+A}\abs{B+B} .
\end{equation}
 
For each $1\le\ell\le k-1$, note that the addition map $\Pi_\ell\times\Pi_{\ell+1}\to\R^2$ given by $(p,q)\mapsto p+q$ is injective as $s_\ell < s_{\ell+1}$.
Therefore
\begin{equation}\label{eq:solymosi-injective}
  \abs{\Pi_\ell+\Pi_{\ell+1}}=\abs{\Pi_\ell}\abs{\Pi_{\ell+1}} \ge 2^{2j}.
\end{equation}
Similarly, as $A\subset\R_{>0}$, the slopes of the points in $\Pi_\ell+\Pi_{\ell+1}$ are all in $(s_\ell,s_{\ell+1})$, so the $\Pi_{\ell}+\Pi_{\ell+1}$ for $1\leq \ell \leq k-1$ are pairwise disjoint.
Finally, $\Pi_\ell+\Pi_{\ell+1}\subseteq(A+A)\times(B+B)$, so disjointness and
\eqref{eq:solymosi-injective} give
\[
  (k-1)2^{2j} \le \sum_{\ell=1}^{k-1}\abs{\Pi_\ell+\Pi_{\ell+1}} = \abs{\bigcup_{\ell=1}^{k-1}\Pi_\ell+\Pi_{\ell+1}}
   \le \abs{A+A}\abs{B+B}.
\]
If $k\ge2$ then $k\le2(k-1)$ and \eqref{eq:solymosi-target} follows.
If $k=1$ then $2^j\le\abs{\Pi_1} = n_{s_1} \le N_0$, so
\[ k2^{2j}=2^{2j}\le\abs{\Pi_1}^2 \leq N_0^2 \leq \abs A\abs B\le\abs{A+A}\abs{B+B} \]
and \eqref{eq:solymosi-target} holds again.
\end{proof}
 
We now convert \cref{thm:asym-solymosi} into a lower bound on the entropy of a product, in terms of the doubling defects $g(P_r)=2\log\abs{P_r}-\log\abs{P_r+P_r}$ of the pieces $P_r$ of the sign-refined dyadic layers of \cref{lem:layers}.
We state the following result in slightly larger generality.
\begin{proposition}\label{prop:h2-product}
Let $\widetilde X$ be zero-free and finitely supported, and let $\{P_r\}_r$ be any
partition of $\operatorname{supp}\widetilde X$ into sign-pure parts whose conditional mass
functions satisfy
\begin{equation}\label{eq:part-flatness}
  p_{P_r}(x) \le \frac2{\abs{P_r}}\qquad\text{for all }x\in P_r,
  \qquad\text{where } p_{P_r}=\operatorname{Law}(\widetilde X\mid\widetilde X\in P_r).
\end{equation}
Then, with
$g(P)=2\log\abs P-\log\abs{P+P}$ and $\widetilde X'$ an independent copy of $\widetilde X$,
\begin{equation}\label{eq:h2-product}
  H(\widetilde X\widetilde X') \ge 2\sum_r\P(\widetilde X\in P_r)\,g(P_r)
  -\log\paren{8\paren{\sum_r \P(\widetilde X\in P_r)\log\abs{P_r}+1}}-4 .
\end{equation}
\end{proposition}
 
\begin{proof}
Let $R,R'$ be the part labels of $\widetilde X,\widetilde X'$, e.g., $\widetilde X \in P_R$.
Write $V_r=\operatorname{Law}(\widetilde X\mid R=r)$, with mass function
$p_{P_r}$ supported on $P_r$, and let $V_t'$ be an independent copy of $V_t$.
Then
\begin{align}
    \notag H(\widetilde X\widetilde X')
    &\ge H(\widetilde X\widetilde X'\mid R,R')
    \\ &=\sum_{r,t}\P(R=r)\P(R'=t) H(\widetilde X\widetilde X'\mid R=r,R'=t) \notag
    \\ &=\sum_{r,t}\P(R=r)\P(R'=t) H(V_r V_t') \notag
    \\ &\geq\sum_{r,t}\P(R=r)\P(R'=t) H_2(V_rV_t'). \label{eq:h2-cells}
\end{align}
To bound $H_2(V_rV_t')$, fix $r,t$ and let
\[ n_\times(z)=\#\{(x,y)\in P_r\times P_t:xy=z\} \qquad\text{for }z\in\R. \]
By \eqref{eq:part-flatness} on both factors,
\[
  \P(V_rV_t'=z)=\sum_{\substack{(x,y)\in P_r\times P_t\\ xy=z}}p_{P_r}(x)p_{P_t}(y)
   \le \frac{4}{\abs{P_r}\abs{P_t}}\,n_\times(z),
\]
and therefore, since $\sum_zn_\times(z)^2=E_\times(P_r,P_t)$ by \eqref{eq:mult-energy}, 
\[ \sum_z\P(V_rV_t'=z)^2 \le \frac{16}{\paren{\abs{P_r}\abs{P_t}}^2}\,E_\times(P_r,P_t), \]
that is,
\begin{align}
    H_2(V_rV_t')
    &= -\log\P(V_rV_t' = V_r'V_t) \notag
    \\ &= -\log\paren{\sum_z\P(V_rV_t'=z)^2} \notag
    \\ &\ge 2\log\abs{P_r}+2\log\abs{P_t}-\log E_\times(P_r,P_t)-4. \label{eq:h2-cell-bound}
\end{align}

By hypothesis, the parts are sign-pure and so $0\notin P_t$.
Then \cref{thm:asym-solymosi} yields
\[
  \log E_\times(P_r,P_t) \le \log\paren{8L_\times(P_r,P_t)}
  +\log\abs{P_r+P_r}+\log\abs{P_t+P_t},
\]
where
\[ L_\times(P_r,P_t)=\floor{\log\min\set{\abs{P_r},\abs{P_t}}}+1. \]
Substituting this into \eqref{eq:h2-cell-bound} and recalling the definition of $g$,
\[ H_2(V_rV_t') \ge g(P_r)+g(P_t)-\log\paren{8L_\times(P_r,P_t)}-4. \]
Returning to \eqref{eq:h2-cells}, this implies
\begin{align}
    H(\widetilde X \widetilde X')
    &\geq -4 + \sum_{r} \P(R=r)g(P_r) + \sum_t\P(R'=t)g(P_t) - \sum_{r,t}\P(R=r)\P(R'=t)\log(8L_\times(P_r,P_t)) \notag
    \\ &= -4 + 2\sum_r \P(R=r)g(P_r) - \sum_{r,t} \P(R=r)\P(R'=t)\log(8L_\times(P_r,P_t)). \label{eq:h2-double}
\end{align}
To bound the double sum,
\[
  L_\times(P_r,P_t) \le 1+\min\set{\log\abs{P_r},\log\abs{P_t}}
   \le 1+\frac12\paren{\log\abs{P_r}+\log\abs{P_t}},
\]
so that by Jensen's inequality,
\begin{align*}
    &\sum_{r,t}\P(R=r)\P(R'=t)\log\paren{8L_\times(P_r,P_t)}
    \\ &\le \log\paren{\sum_{r,t}\P(R=r)\P(R'=t)\, 8\paren{1+\frac12\log\abs{P_r}+\frac12\log\abs{P_t}}}
    \\ &= \log\paren{8\paren{1+\sum_r \P(R=r)\log\abs{P_r}}}.
\end{align*}
Substituting this into \eqref{eq:h2-double} yields \eqref{eq:h2-product}.
\end{proof}
The same argument bounds the entropy of a quotient, which will also be needed in \cref{sec:eer}.
\begin{proposition}\label{prop:h2-ratio}
Under the same assumptions and notation as \cref{prop:h2-product},
\[ H(\widetilde X/\widetilde X') \ge 2\sum_r\P(\widetilde X\in P_r)\,g(P_r)-\log\paren{8\paren{\sum_r \P(\widetilde X\in P_r)\log\abs{P_r}+1}}-4 .\]
\end{proposition}
 
\begin{proof}
The quotient is well-defined as $\widetilde X'\ne0$.
Fix parts $P_r,P_t$ and, for $z\ne0$, similar to before let
\[ n(z)=\#\{(x,y)\in P_r\times P_t:\ x/y=z\}. \]
Then the proof proceeds identically with $n(z)$ in place of $n_\times(z)$, where in particular
\[ \sum_z n(z)^2 = \#\set{(x,y,x',y')\in(P_r\times P_t)^2:\ \frac xy=\frac{x'}{y'}}
  =\#\{(x,y,x',y'):\ xy'=x'y\}=E_\times(P_r,P_t). \qedhere \]
\end{proof}

\section{Small-doubling linear entropic Elekes--Ruzsa theorem}\label{sec:eer} 
In this section we prove the finite-support special case of \cref{thm:eer} in \cref{cor:eer-eff}, which we extend to the general case in \cref{sec:main}.
Given \cref{prop:h2-product}, which lower bounds the product entropy in terms of the doublings of the sign-refined dyadic layers of \cref{lem:layers}, our remaining task is to bound these doublings via the entropic doubling of the overall random variable.
The main technical lemma is \cref{lem:gibbs}, a doubling estimate for an arbitrary subset of a dyadic layer.
Throughout this section, $\widetilde X$ will denote a zero-free finitely supported real-valued random variable with $\widetilde H=H(\widetilde X)\ge2$, and
$\widetilde X'$ denotes an independent copy; recall
\[
  \widetilde u=\frac{H(\widetilde X+\widetilde X')}{\widetilde H}-1\in[0,1],
  \qquad
  \widetilde m=\frac{H(\widetilde X\widetilde X')}{\widetilde H} .
\]
We fix the sign-refined dyadic layers of \cref{lem:layers}(e): a
partition $\set{A_i}_{i\in\mc I}$ of $\operatorname{supp}\widetilde X$ into sign-pure sets,
indexed by a finite set $\mc I$, with $I$ the ($\widetilde X$-measurable) random variable determined by $\widetilde X\in A_I$, and for brevity let
\[
  w_i=\P(I=i),\qquad a_i=\log\abs{A_i},\qquad
  V_i=\mathrm{Law}(\widetilde X\mid I=i)\ \text{with probability mass function }p_i .
\]
Layers of zero weight are discarded from $\mc I$, so $w_i>0$ for all $i\in\mc I$.
Only four consequences of \cref{lem:layers} are used below: the two-sided flatness
\begin{equation}\label{eq:flat}
  \frac1{2\abs{A_i}} \le p_i(x) \le \frac2{\abs{A_i}}\qquad\text{for all }x\in A_i,
\end{equation}
of part (a); the inherited flatness
\[ \max p_D<2\min p_D\le\frac{2}{\abs D} \]
for the conditional probability mass function of any
nonempty $D\subseteq A_i$, from part (d); the label-entropy bound
\begin{equation}\label{eq:hI}
  H(I) \le \log(1+\widetilde H)+\log e+1
\end{equation}
of part (e); and the two-sided estimate
\begin{equation}\label{eq:abar}
  \widetilde H-\log(1+\widetilde H)-\log e-1 \le \widetilde H-H(I) \le \overline a:=\sum_{i\in\mc I}w_ia_i \le \widetilde H+1
\end{equation}
of part (c). Note that \eqref{eq:flat} gives $H_\infty(V_i)\ge a_i-1$.
 
\subsection{Bounding excess entropy}\label{ss:budgets}
In this subsection, as the layers of the dyadic decomposition are flat, the entropy of a sum of two of them will be measured against the reference $\max\set{a_i,a_j}$, a heuristic proxy for its entropy, and the total excess over this reference, summed across all pairs of dyadic layers, will be bounded by $\widetilde u\widetilde H$.
 
\begin{definition}\label{def:budgets}
For an ordered pair $(i,j)\in\mc I\times\mc I$, with $V_j'$ an independent copy of $V_j$, let
\[
  e^\pm_{ij}=H(V_i\pm V_j')-\max\set{a_i,a_j}+1,
\]
by which we mean
\[ e_{ij}^+ = H(V_i+V_j') - \max\set{a_i,a_j} + 1 \qquad\text{and}\qquad e_{ij}^- = H(V_i-V_j') - \max\set{a_i,a_j}+1, \]
with all statements involving $e_{ij}^\pm$ holding for both cases.
The $e^\pm_{ij}$ are nonnegative by (E5) of \cref{lem:entropy-facts} along with
$H(V_i)\ge H_\infty(V_i)\ge a_i-1$ and likewise for $j$:
\[ e_{ij}^\pm \geq \max\set{H(V_i),H(V_j)}-\max\set{a_i,a_j}+1 \geq \max\set{H_\infty(V_i),H_\infty(V_j)}-\max\set{a_i,a_j}+1 \geq 0. \]
\end{definition}
 
\begin{lemma}\label{lem:budgets}
Let $\beta\ge1$. Then
\begin{equation}\label{eq:budgets}
  \sum_{i,j}w_iw_j\paren{\beta(e^+_{ij}+e^-_{ij})+\abs{a_i-a_j}}
   \le \beta\paren{3\widetilde u\widetilde H+2H(I)+2} .
\end{equation}
\end{lemma}
 
\begin{proof}
Let $I$ and $I'$ be the layer labels of $\widetilde X$ and $\widetilde X'$; they are
independent with the same law, and conditionally on $\set{I=i,I'=j}$ the pair
$(\widetilde X,\widetilde X')$ is distributed as $(V_i,V_j')$. Hence,
\[
  H(\widetilde X\pm\widetilde X') \ge H(\widetilde X\pm\widetilde X'\mid I,I')
   = \sum_{i,j}w_iw_jH(V_i\pm V_j') .
\]
Since
\[ \max\set{a_i,a_j}=\frac12(a_i+a_j)+\frac12\abs{a_i-a_j}, \]
averaging gives
\[
  \sum_{i,j}w_iw_j\max\set{a_i,a_j}=\overline a+\frac12\sum_{i,j}w_iw_j\abs{a_i-a_j} ,
\]
so, by the definition of $e^\pm_{ij}$ and then \eqref{eq:abar},
\[
  \sum_{i,j}w_iw_je^\pm_{ij}+\frac12\sum_{i,j}w_iw_j\abs{a_i-a_j}
   \le H(\widetilde X\pm\widetilde X')-\overline a+1
   \le H(\widetilde X\pm\widetilde X')-\widetilde H+H(I)+1 .
\]
Now $H(\widetilde X+\widetilde X')=(1+\widetilde u)\widetilde H$ by definition, while
$H(\widetilde X-\widetilde X')\le2H(\widetilde X+\widetilde X')-\widetilde H
=(1+2\widetilde u)\widetilde H$ by \eqref{eq:difference-from-sum}.
Adding the two instances of the last display, one for each sign,
\[
  \sum_{i,j}w_iw_j\bracket{e^+_{ij}+e^-_{ij}+\abs{a_i-a_j}}
   \le 3\widetilde u\widetilde H+2H(I)+2 .
\]
Since $\beta\ge1$ and $\abs{a_i-a_j}\ge0$, the left-hand side of \eqref{eq:budgets} is at most
$\beta$ times the left-hand side of the last display.
\end{proof}
For each $i\in\mc I$, picking $j\in\mc I$ to minimize the summand immediately implies the following, which will be useful for our main technical lemma, \cref{lem:gibbs}.
\begin{corollary}\label{cor:partner}
Let $\beta\ge1$. There is a map $j(\cdot):\mc I\to\mc I$ with
\begin{equation}\label{eq:partner}
  \sum_{i\in\mc I}w_i\bracket{\beta(e^+_{i,j(i)}+e^-_{i,j(i)})+\abs{a_i-a_{j(i)}}}
   \le \beta\paren{3\widetilde u\widetilde H+2H(I)+2} .
\end{equation}
\end{corollary}
For the rest of \cref{sec:eer} we fix the partner map $j(\cdot)$ of \cref{cor:partner}. 
\begin{definition}\label{def:deficits}
Given a layer $i$ with partner $j=j(i)$, write
\[ r_+(s)=\P\paren{V_i+V_j'=s},\qquad r_-(t)=\P\paren{V_i-V_j'=t}; \]
both depend on the pair $(i,j)$, which is suppressed from the notation for simplicity.
Thus $r_+$ is a probability measure on $A_i+A_j$ and $r_-$ is a probability measure on $A_i-A_j$, with full support on those sets.

For $x\in A_i$, recalling $p_j$ denotes the probability mass function of $\widetilde X|I=j$, define
\[
  d_\pm(x)=\sum_{z\in A_j}p_j(z)\log\frac1{r_\pm(x\pm z)}-\max\set{a_i,a_j}+1.
\]
Equivalently, for $Z\sim V_j$,
\begin{equation}\label{eq:defdef}
  \E\bracket{\log\frac1{r_\pm(x\pm Z)}} = \max\set{a_i,a_j}-1+d_\pm(x).
\end{equation}
In essence, as opposed to $e_{ij}^\pm$ measuring the excess entropy of $V_i\pm V_j'$, the function $d_\pm(x)$ for $x\in A_i=\operatorname{supp} V_i$ measures the contribution to the excess entropy from the case $V_i=x$.
\end{definition}
The following result provides pointwise control on $d_\pm$, and its point (b) makes rigorous the previous observation.
\begin{lemma}\label{lem:deficit-props}
Let $i\in\mc I$ and $j=j(i)$. Then:
\begin{enumerate}
\item[(a)] $r_\pm\le2^{1-\max\set{a_i,a_j}}$ everywhere; consequently $d_\pm(x)\ge0$
pointwise for every $x\in A_i$;
\item[(b)] $\E[d_\pm(V_i)]=e^\pm_{i,j(i)}$;
\item[(c)] $0\le d_\pm(x)\le3+\min\set{a_i,a_j}$ for every $x\in A_i$.
\end{enumerate}
\end{lemma}
 
\begin{proof}
(a) By \eqref{eq:flat}, for any $s\in A_i+A_j$,
\[ r_+(s)=\sum_xp_i(x)p_j(s-x)\le\max p_i\le\frac{2}{\abs{A_i}} \]
and similarly $r_+(s)\leq\frac{2}{\abs{A_j}}$, so
\[ r_+\le\frac{2}{\max\set{\abs{A_i},\abs{A_j}}}=2^{1-\max\set{a_i,a_j}}; \]
an essentially identical computation gives the same bound for $r_-$.
Hence
\[ \log\frac{1}{r_\pm}\ge\max\set{a_i,a_j}-1 \]
pointwise, from which we find $d_\pm\ge0$ pointwise.
 
(b) We directly compute
\[
  \E\bracket{\sum_{z\in A_j}p_j(z)\log\frac1{r_+(V_i+z)}}
  =\sum_{x,z}p_i(x)p_j(z)\log\frac1{r_+(x+z)}
  =\sum_{s}r_+(s)\log\frac1{r_+(s)}=H(V_i+V_j') ,
\]
so
\[ \E[d_+(V_i)]=H(V_i+V_j')-\max\set{a_i,a_j}+1=e^+_{i,j(i)}; \]
the proof for $d_-$ is the same.
 
(c) The lower bound is (a). For the upper bound, the lower half of \eqref{eq:flat}
gives
\[ r_+(x+z)\ge p_i(x)p_j(z)\ge\frac1{4\abs{A_i}\abs{A_j}} \]
for $x\in A_i$, $z\in A_j$, so $\log(1/r_+(x+z))\le2+a_i+a_j$ and
\[ d_+(x) \le 2+a_i+a_j-\max\set{a_i,a_j}+1 = 3+\min\set{a_i,a_j}. \]
The proof for $d_-$ is the same.
\end{proof}
 
\subsection{Level sets, and localized doubling}\label{ss:levelsets}
This subsection proves the main technical lemma, \cref{lem:gibbs}, which bounds the doubling of an arbitrary subset of a layer $A_i$ in terms of the maximum values of $d_\pm$ on this subset.
As the previous subsection bounds the expected value of $d_\pm$ across $A_i$ via $e_{ij}^\pm$, whose expected value across $i\in\mc I$ is bounded by the overall entropic doubling $\widetilde u$ by \cref{cor:partner}, one can imagine further partitioning $A_i$ into $\beta$-adic intervals for some $\beta>1$ based on the values of $d_\pm$, bounding the doubling terms for each interval via \cref{lem:gibbs}, and recombining, where the contributions of the more dangerous terms, i.e., with high $d_\pm$ values, are controlled by a tail bound on the probability of such an event occurring.
This is essentially what we do, where this $\beta$ is the same $\beta$ as we added without explanation to \cref{lem:budgets}, we control the loss across the $\beta$-adic intervals in \cref{lem:levelsets}, and then recombine in \cref{prop:perlayer}.

\begin{definition}\label{def:levels}
Fix $\beta>1$, a layer $i$ and its partner $j=j(i)$. The \emph{level-set partition} $\mc D_i$ of $A_i$ is the partition in which each part consists of the $x\in A_i$ that share the same value of the ordered pair
\[
  \paren{\ceil{\log_\beta\max\set{d_+(x),1}},\ceil{\log_\beta\max\set{d_-(x),1}}}.
\]
\end{definition}
 
\begin{lemma}\label{lem:levelsets}
Fix $\beta>1$ and a layer $i$, with partner $j=j(i)$. Then $\mc D_i$ is a partition of $A_i$ into
at most $\paren{2+\log_\beta(3+a_i)}^2$ nonempty parts, and
\[
  \sum_{D\in\mc D_i}p_i(D)\,\sup_{x\in D}d_\pm(x) \le 1+\beta e^\pm_{i,j(i)} .
\]
\end{lemma}
 
\begin{proof}
Clearly $\mc D_i$ partitions $A_i$, so $\sum_{D\in\mc D_i}p_i(D)=1$.
By
\cref{lem:deficit-props}(c) the quantity $\max\set{d_\pm(x),1}$ lies in
$[1,3+\min\set{a_i,a_j}]$, so each coordinate of the map takes at most
\[ 1+\ceil{\log_\beta(3+\min\set{a_i,a_j})}\le2+\log_\beta(3+a_i) \]
values; squaring bounds the number of nonempty parts.
 
Fix a part $D\in\mc D_i$ and let $a$ be the common value of $\ceil{\log_\beta\max\set{d_+,1}}$ on it, so
\[ \sup_{x\in D}d_+(x)\le\beta^a. \]
For a lower bound, notice that $\beta^a\le1+\beta d_+(x)$ pointwise on $D$, where the 1 term ensures the case $a=0$ holds, and otherwise for $a\geq 1$ we have $d_+(x)>\beta^{a-1}$, which implies the bound.
Averaging over the parts,
\begin{align*}
  \sum_{D\in\mc D_i}p_i(D)\,\sup_{x\in D}d_+(x)
   &\le \sum_{D\in\mc D_i}\sum_{x\in D}p_i(x)\paren{1+\beta d_+(x)}
   \\ &= 1+\beta\,\E[d_+(V_i)]
   \\ &= 1+\beta e^+_{i,j(i)}
\end{align*}
by \cref{lem:deficit-props}(b).
The proof is identical for $d_-$.
\end{proof}
 
As promised, we now bound $\log\abs{D+D}$ for an arbitrary nonempty
$D\subseteq A_i$, in terms of the flat reference level of the layer and the suprema over $D$ of $d_\pm$.
 
\begin{lemma}\label{lem:gibbs}
Let $i\in\mc I$, let $j=j(i)$ be its partner, and let $\emptyset\ne D\subseteq A_i$. Then
\begin{equation}\label{eq:gibbs}
  \log\abs{D+D} \le a_i+\abs{a_i-a_j}-1+\sup_{x\in D}d_+(x)+\sup_{y\in D}d_-(y) .
\end{equation}
\end{lemma}
 
\begin{proof}
Let $S\sim\Unif(D+D)$ and for each $s\in D+D$ fix one representation $s=x_s+y_s$ with $x_s,y_s\in D$.
Let $Z\sim V_j$ be independent of $S$, and set $U=x_S+Z$ and $V=y_S-Z$.
Given $S$, and thus the fixed pair $(x_S,y_S)$, the map $Z\mapsto(U,V)$ is bijective, so
\[ H\paren{U,V\mid S}=H\paren{Z\mid S}=H(Z) \]
by independence of $Z$ and $S$.
Since $U+V=S$, note that $S$ is a function of $(U,V)$, so
\begin{equation}\label{eq:witness}
  H(U,V) = H(S)+H\paren{U,V\mid S} = \log\abs{D+D}+H(Z).
\end{equation}
We note that this first step of the proof, i.e., the definitions of $U$ and $V$, is inspired by the standard proof of the Ruzsa triangle inequality \cite{Ruzsa1978}.
 
The variable $U$ takes values in $D+A_j\subseteq A_i+A_j$, on which $r_+$, which recall is the probability mass function of $V_i+V_j'$, is a strictly positive probability measure; its restriction to the range of $U$ is therefore a sub-probability measure, strictly positive on
$\operatorname{supp}U$, so Gibbs' inequality ((E6) from \cref{lem:entropy-facts}) and
\eqref{eq:defdef} imply
\begin{align*}
    H(U) &\le \E\bracket{\log\frac1{r_+(U)}}
    \\ &= \E_S\,\E_Z\bracket{\log\frac1{r_+(x_S+Z)}}
    \\ &= \E_S\bracket{\max\set{a_i,a_j}-1+d_+(x_S)}
    \\ &\le \max\set{a_i,a_j}-1+\sup_{x\in D}d_+(x).
\end{align*}
Identically, $V$ takes values in $D-A_j\subseteq A_i-A_j$, on which $r_-$ is a strictly
positive probability measure, so
\[ H(V)\le\max\set{a_i,a_j}-1+\sup_{y\in D}d_-(y). \]

Thus, by \eqref{eq:witness},
\[
  \log\abs{D+D}+H(Z) = H(U,V) \le H(U)+H(V)
   \le 2\max\set{a_i,a_j}-2+\sup_{x\in D}d_+(x)+\sup_{y\in D}d_-(y) .
\]
Recall by \eqref{eq:flat} we have $H(Z)\ge H_\infty(V_j)\ge a_j-1$.
Therefore
\[ \log\abs{D+D} \le 2\max\set{a_i,a_j}-a_j-1+\sup_{x\in D}d_+(x)+\sup_{y\in D}d_-(y). \]
Observing that
\[ 2\max\set{a_i,a_j}-a_j-1=a_i+\abs{a_i-a_j}-1 \]
completes the proof.
\end{proof}
We now translate this to a lower bound on the doubling defect $g(D)$.
\begin{corollary}\label{cor:perpart}
Let $i\in\mc I$ and let $\emptyset\ne D\subseteq A_i$.
Then, with $g(P)=2\log\abs{P}-\log\abs{P+P}$,
\[
    g(D) \ge a_i-\abs{a_i-a_{j(i)}}-\sup_{x\in D}d_+(x)-\sup_{y\in D}d_-(y)
    -2\log\frac1{p_i(D)}-1 .
\]
\end{corollary}
\begin{proof}
By \eqref{eq:flat}, $p_i(D)\le2\abs D/\abs{A_i}$, so
$2\log\abs D\ge2a_i+2\log p_i(D)-2$.
\end{proof}
We now recombine across the level-set partition to return to the entire dyadic layer $A_i$.
\begin{proposition}\label{prop:perlayer}
Let $\beta>1$ and let $\mc D_i$ be the level-set partition of the layer $A_i$. Then
\[
  \sum_{D\in\mc D_i}p_i(D)\,g(D)
   \ge a_i-\beta\paren{e^+_{i,j(i)}+e^-_{i,j(i)}}-\abs{a_i-a_{j(i)}}
   -4\log\paren{2+\log_\beta(3+a_i)}-3 .
\]
\end{proposition}
\begin{proof}
Apply \cref{cor:perpart} to each part $D\in\mc D_i$ and average against $p_i(D)$:
\[
  \sum_{D\in\mc D_i}p_i(D)\,g(D)
   \ge a_i-\abs{a_i-a_{j(i)}}
   -\sum_{D\in\mc D_i}p_i(D)\paren{\sup_Dd_++\sup_Dd_-}
   -2\sum_{D\in\mc D_i}p_i(D)\log\frac1{p_i(D)}-1 .
\]
By \cref{lem:levelsets} the first of the two sums is at most
$2+\beta\paren{e^+_{i,j(i)}+e^-_{i,j(i)}}$, while by (E4) of \cref{lem:entropy-facts} the second sum is bounded by
\[ \sum_{D\in\mc D_i}p_i(D)\log\frac{1}{p_i(D)} \leq \log\abs{\mc D_i}\leq2\log\paren{2+\log_\beta(3+a_i)}, \]
with the second inequality following from \cref{lem:levelsets}.
\end{proof}
\begin{remark}\label{rmk:shape}
By \cref{prop:perlayer} the whole cost of the layer $i$ is
$\beta(e^+_{i,j(i)}+e^-_{i,j(i)})$ plus terms of
$O(\log\log)$ type in entropy.
\cref{cor:partner} turns $e^+_{i,j(i)}+e^-_{i,j(i)}$ into $3\widetilde u\widetilde H$, where from the proof of \cref{lem:budgets}, controlling the sum and difference entropies contribute 1 and 2, respectively, to this coefficient of 3.
\cref{prop:h2-product} then doubles it, yielding a coefficient of 6 decomposed as
\[
  C = \underbrace{2}_{\text{Solymosi}}\times(\underbrace{1}_{\text{sum budget}}
  +\underbrace{2}_{\text{difference budget}}) = 6.
\]
This will impact the coefficient of 6 in \cref{thm:eer}, and thus the value $\delta=\frac17$; see \cref{prop:master}.
\end{remark}
 
\subsection{Assembling together}\label{ss:eer-thm}
We now assemble our previous work to prove the following technical version of the linear entropic Elekes--Ruzsa theorem (still for finite support).
\begin{theorem}\label{thm:eer-full}
Let $\widetilde X$ be zero-free and finitely supported with $\widetilde H=H(\widetilde X)>0$, and let $\beta>1$.
Then, using the notation of $\widetilde u$ and $\widetilde m$ from \cref{thm:eer},
\begin{equation}\label{eq:eer-full}
  \widetilde m \ge 2-6\beta\,\widetilde u-\frac{\mathrm{Err}(\beta,\widetilde H)}{\widetilde H},
\end{equation}
where
\begin{align}\label{eq:Err}
  \mathrm{Err}(\beta,\widetilde H) &= (4\beta+2)\paren{\log(1+\widetilde H)+\log e+1}+4\beta
  \\ &\quad +8\log\paren{2+\log_\beta(4+\widetilde H)} + \log\paren{8(\widetilde H+2)}+10. \notag
\end{align}
\end{theorem}
 
\begin{proof}
Use the sign-refined layers $\set{A_i}_{i\in\mc I}$ of the standing setup, the partner map $j(\cdot)$ of \cref{cor:partner} for this $\beta$, and for each layer the level-set partition $\mc D_i$ of \cref{def:levels} for this $\beta$.
 
The resulting collection $\bigcup_{i\in\mc I}\mc D_i$ of nonempty parts partitions
$\operatorname{supp}\widetilde X$: each part lies in a single layer $A_i$, hence is sign-pure, and its conditional probability mass function obeys $p_D\le2/\abs D$ by \cref{lem:layers}.
Thus \cref{prop:h2-product} applies, with the weight $w_ip_i(D)$ attached to $D\in\mc D_i$, and with
\[
    \sum_{i\in\mc I}w_i\sum_{D\in\mc D_i}p_i(D)\log\abs D
    \le \sum_{i\in\mc I}w_ia_i=\overline a \le \widetilde H+1
\]
by \eqref{eq:abar}, since $\abs D\le\abs{A_i}$ and $\sum_{D\in\mc D_i}p_i(D)=1$.
It yields
\begin{equation}\label{eq:apply-sol}
  H(\widetilde X\widetilde X') \ge
  2\sum_{i\in\mc I}w_i\sum_{D\in\mc D_i}p_i(D)\,g(D)
  -\log\paren{8(\widetilde H+2)}-4 .
\end{equation}
By \cref{prop:perlayer}, the inner double sum obeys
\[
  \sum_{i\in\mc I}w_i\sum_{D\in\mc D_i}p_i(D)\,g(D)
   \ge \overline a-\sum_{i\in\mc I}w_i\bracket{\beta\paren{e^+_{i,j(i)}+e^-_{i,j(i)}}
   +\abs{a_i-a_{j(i)}}+4\log\paren{2+\log_\beta(3+a_i)}+3}
\]
and $\overline a\ge\widetilde H-H(I)$ by \eqref{eq:abar}.
Here \cref{cor:partner} bounds the first two terms of the bracket: 
\[ \sum_{i\in\mc I}w_i\bracket{\beta\paren{e^+_{i,j(i)}+e^-_{i,j(i)}}+\abs{a_i-a_{j(i)}}} \leq \beta\paren{3\widetilde u\widetilde H+2H(I)+2}. \]
As $a\mapsto\log\paren{2+\log_\beta(3+a)}$ is increasing and concave on $[0,\infty)$, by Jensen's inequality and recalling $\overline a\le\widetilde H+1$,
\[ \sum_iw_i\,4\log\paren{2+\log_\beta(3+a_i)} \le 4\log\paren{2+\log_\beta(4+\widetilde H)}. \]
Substituting into \eqref{eq:apply-sol},
\[
  H(\widetilde X\widetilde X') \ge 2\widetilde H-6\beta\widetilde u\widetilde H
  -\bracket{(4\beta+2)H(I)+4\beta
  +8\log\paren{2+\log_\beta(4+\widetilde H)}
  +\log\paren{8(\widetilde H+2)}+10}.
\]
Since $4\beta+2>0$, the bound \eqref{eq:hI} on $H(I)$ may now
be inserted, and dividing by $\widetilde H$ gives the result.
\end{proof}
\begin{remark}
    The intercept $2$ in \eqref{eq:eer-full} is optimal: by subadditivity,
    $H(\widetilde X\widetilde X')\le H(\widetilde X,\widetilde X')=2\widetilde H$, so
    $\widetilde m\le2$ always, and \cref{thm:eer-full} is asymptotically exact as
    $\widetilde u\to0$ and $\widetilde H\to\infty$.
\end{remark}
We now optimize $\beta$ at $\beta=1+1/\widetilde H$, which balances the inflation $6\beta$ of the coefficient against the term $8\log(2+\log_\beta(4+\widetilde H))$, and yields a total error of order $\log\widetilde H/\widetilde H$.
\begin{corollary}\label{cor:eer-eff}
Let $\widetilde X$ be zero-free and finitely supported with $\widetilde H\ge4$, and take
$\beta=1+1/\widetilde H$ in \cref{thm:eer-full}. Then, using the notation of $\widetilde u$ and $\widetilde m$ from \cref{thm:eer},
\begin{equation}\label{eq:eer-beta}
  \widetilde m \ge 2-\paren{6+\frac6{\widetilde H}}\widetilde u
  -\frac{18\log\widetilde H+51}{\widetilde H} .
\end{equation}
Consequently, with $F_2(\widetilde u)=2-6\widetilde u$ and
\begin{equation}\label{eq:vartheta}
  \vartheta(\eta)=\begin{cases}\frac{18\log\eta+57}{\eta}&\text{if }\eta\ge4, \\ 
  \vartheta(4)=\frac{93}4 & \text{if }0<\eta<4, \end{cases}
\end{equation}
one has
\begin{equation}\label{eq:eer-clean}
  \widetilde m \ge F_2(\widetilde u)-\vartheta(\widetilde H)
\end{equation}
for every zero-free finitely supported $\widetilde X$ with $\widetilde H>0$.
$\vartheta:(0,\infty)\to(0,\infty)$ is a nonincreasing function.
\end{corollary}
\begin{proof}
Write $\eta=\widetilde H\ge4$ and $t=1/\eta\in(0,\frac14]$, so $\beta=1+t$.
Since $\widetilde u\le1$ by (E5$'$) of \cref{lem:entropy-facts},
\[ 6\beta\widetilde u=6\widetilde u+6t\widetilde u\le6\widetilde u+6/\eta, \]
so \eqref{eq:eer-full} implies \eqref{eq:eer-clean} with $\vartheta$ replaced by
$\paren{6+\mathrm{Err}(1+t,\eta)}/\eta$.
Thus, \eqref{eq:eer-beta} and the case $\eta\geq4$ of \eqref{eq:eer-clean} follow from
\begin{equation}\label{eq:maj-goal}
  \mathrm{Err}(1+t,\eta) \le 18\log\eta+51,
\end{equation}
which we now confirm.
Note the case $0<\eta<4$ of \eqref{eq:eer-clean} is trivial, since $\vartheta(\eta)=\frac{93}4>2\ge F_2(\widetilde u)$ for every
$\widetilde u\ge0$ while $\widetilde m\ge0$.
 
As $\log(1+t)\ge t$ for $0\le t\le1$, recalling $\log$ denotes the base-2 logarithm, $\log\beta\ge t=1/\eta$ so
\begin{equation}\label{eq:logbeta}
  \log_\beta(4+\eta) \le \eta\log(4+\eta).
\end{equation}
As $\eta\ge4$ and thus $1+\eta\le\frac54\eta$,
\begin{equation}\label{eq:log1eta}
  \log(1+\eta) \le \log\eta+\log\frac54.
\end{equation}
Since $\beta=1+1/\eta\le\frac54$, 
\[
  (4\beta+2)\paren{\log(1+\eta)+\log e+1}+4\beta
   \le 7\paren{\log\eta+\log\paren{\frac54e}+1}+5 \approx 7\log\eta+24.3525 .
\]
By \eqref{eq:logbeta} and $4+\eta\le2\eta$, we have $\log_\beta(4+\eta)\le\eta(1+\log\eta)$.
Since $\eta(1+\log\eta)\geq12$, we absorb the additive constant of 2 via 
\[ 2+\log_\beta(4+\eta) \le \frac76\,\eta\,(1+\log\eta), \]
so
\[
    8\log\paren{2+\log_\beta(4+\eta)} \le 8\log\frac76+8\log\eta+8\log(1+\log\eta) \approx 8\log\eta+8\log(1+\log\eta)+1.7792.
\]
The function $\phi(L)=8\log(1+L)-2L$ has $\phi'(L)=8\log e/(1+L)-2$, so on $[2,\infty)$ it is
maximized at $1+L=4\log e\approx5.77078$, where
$\phi\le8\log(4\log e)-2(4\log e-1)\approx10.6886<10.69$; hence
$8\log(1+L)\le2L+10.69$ and
\[ 8\log\paren{2+\log_\beta(4+\eta)} \le 10\log\eta+12.47. \]
Lastly, as $\eta+2\leq\frac32\eta$,
\[ \log\paren{8(\eta+2)}\le\log\eta+3+\log\frac32\approx\log\eta+3.5850. \]
Combining these bounds yields
\begin{align*}
  \mathrm{Err}(1+t,\eta)
  & \le 18\log\eta
  +\paren{24.3525+12.47+3.5850+10}\\
  & < 18\log\eta+51,
\end{align*}
which is \eqref{eq:maj-goal}.

Lastly, on $[4,\infty)$,
\[
  \frac{d}{d\eta}\,\frac{18\log\eta+57}{\eta}
   = \frac{18\log e-\paren{18\log\eta+57}}{\eta^2} < 0,
\]
so $\vartheta$ is nonincreasing.
\end{proof}
\begin{remark}\label{rmk:quotient}
    \Cref{thm:eer-full} and \cref{cor:eer-eff} hold verbatim with
    $H(\widetilde X\widetilde X')$ replaced by $H(\widetilde X/\widetilde X')$, hence with
    $\widetilde m$ replaced by $H(\widetilde X/\widetilde X')/\widetilde H$: the product enters the proof only through the application of \cref{prop:h2-product} in \eqref{eq:apply-sol}, and \cref{prop:h2-ratio} supplies the same inequality for the quotient with the same constants; every other step bounds $\log\abs{D+D}$ and is unchanged.
    This may be of independent interest.
\end{remark}
The extension of \cref{cor:eer-eff} to $\widetilde H < 4$ is trivial, but we extend $\vartheta$ as such for convenience later in \cref{prop:master}.
\section{The incidence-side lower bound}\label{sec:incidence}
\cref{thm:GGK_1_3} was unable to avoid the obstacle of examples where $H_\infty(X)=O(1)$ but $H(X)$ is arbitrarily large.
In short, this section uses the dyadic decomposition of \cref{lem:layers} so that each piece $V_i$ of the partition is flat, so that $H_\infty(V_i)\approx H(V_i)$.
This causes a number of complications compared to the proof of \cref{thm:GGK_1_3}, by Gavalakis, Goh, and Kontoyiannis \cite{GGK} (one may note that simply replacing $H_\infty(X)$ with $H(X)$ in \cref{thm:GGK_1_3} would have proven \cref{conjecture:main} with $\delta=\frac16$, yet our result only establishes $\delta=\frac17$) and so, similar to \cref{sec:eer}, great care must be taken to control the resulting errors caused by this decomposition process.
In both \cref{sec:incidence,sec:upper}, we observed some minor errors in the arguments of Gavalakis, Goh, and Kontoyiannis \cite{GGK}, though we believe their conclusions remain intact; we remark upon these when our arguments reach the analogous points, and in an abundance of caution our proofs are provided completely, i.e., do not rely on any of their results.

\Cref{subsec:expander} proves a lower bound
\[ H(U_{G_1}(U_{G_2}+U_{G_3}))\geq\frac12\log\paren{\abs{G_1}\abs{G_2}\abs{G_3}}-K_R \] 
for arbitrary subsets $G_1,G_2,G_3\subset\R$, with some additional conditions, but in particular with no relation imposed on the three sizes $\abs{G_1},\abs{G_2},\abs{G_3}$.
The balanced case $\abs{G_1}=\abs{G_2}=\abs{G_3}$ was previously established by Gavalakis, Goh, and Kontoyiannis \cite[Lemma 3.2]{GGK}, but this generalized version, whose proof is slightly more involved, is needed for \cref{subsec:layered}.
\cref{subsec:layered} extends this lower bound to general random variables via the dyadic decomposition of \cref{lem:layers} (exact uniformity will be obtained by applying the uniformization lemma, \cref{lem:vadhan}, on the dyadic layers), and because the dyadic layers can have very different sizes, we thus need the lower bound in the general, unbalanced case.
\subsection{The product-sum lower bound}\label{subsec:expander}
\begin{theorem}\label{thm:expander}
Let $G_1\subset\R\setminus\{0\}$ and $G_2,G_3\subset\R$ be nonempty finite
sets, and let $U_{G_1},U_{G_2},U_{G_3}$ be independent uniform random variables on
$G_1,G_2,G_3$. We assume
\begin{equation}\label{eq:nondeg}
  \P\paren{U_{G_2}+U_{G_3}=0}<1,
\end{equation}
or equivalently exclude the degenerate case that $G_2=\set{b}$ and $G_3=\set{-b}$ for some $b\in\R$.
Then, with $K_R=2+\log C_{R}$ as in \cref{def:KR},
\[
  H\paren{U_{G_1}(U_{G_2}+U_{G_3})} \ge
  \frac12\log\paren{\abs{G_1}\,\abs{G_2}\,\abs{G_3}}-K_R.
\]
\end{theorem}
Throughout this subsection we assume the setting of \cref{thm:expander}, and define
\[
  \mathcal N=\abs{G_1}\abs{G_2}\abs{G_3},\qquad S=U_{G_2}+U_{G_3},\qquad W=U_{G_1}S,
  \qquad\text{and}\qquad \theta=\P(S=0)<1.
\]
Also let
\[ n_0=\#\{(b,c)\in G_2\times G_3:b+c=0\}, \]
so that $\theta=\frac{n_0}{\abs{G_2}\abs{G_3}}$.
The proof splits into casework on $\abs{G_1}>\abs{G_2}\abs{G_3}$ versus $\abs{G_1}\le\abs{G_2}\abs{G_3}$, the former being resolved via an elementary argument and the latter relying on Rudnev's incidence bound, \cref{thm:incidence}.

We first record some basic observations about $\theta$.
\begin{lemma}\label{lem:theta}
With the notation above,
\begin{enumerate}
\item[\textup{(P1)}] $\theta\le1/\max\set{\abs{G_2},\abs{G_3}}$;
\item[\textup{(P2)}] $\theta\le\frac12$;
\item[\textup{(P3)}] If $\theta>0$ then $\abs{G_2}\abs{G_3}\le\theta^{-2}$.
\end{enumerate}
\end{lemma}
 
\begin{proof}
(P1) For each $b\in G_2$ there is at most one $c\in G_3$ with $b+c=0$, so $n_0\le\abs{G_2}$ and symmetrically $n_0\le\abs{G_3}$, thus
\[
  \theta\le \frac{\min\set{\abs{G_2},\abs{G_3}}}{\abs{G_2}\abs{G_3}}
  =\frac1{\max\set{\abs{G_2},\abs{G_3}}}.
\]

(P2) If $\max\set{\abs{G_2},\abs{G_3}}\ge2$ then (P1) gives $\theta\le\frac12$; otherwise $\abs{G_2}=\abs{G_3}=1$ and so $\theta<1$ implies $b+c\ne0$, so $\theta=0$.
 
(P3) By (P1),
\[ \abs{G_2}\abs{G_3}\le\max\set{\abs{G_2},\abs{G_3}}^2\le\theta^{-2}. \qedhere \]
\end{proof}
 
\subsubsection{The case $\abs{G_1}>\abs{G_2}\abs{G_3}$}
We first record a simple bound on $H(W)$.
\begin{lemma}\label{lem:regime1-base}
For any nonempty finite $G_1\subseteq\R\setminus\{0\}$ and $G_2,G_3\subseteq\R$,
\[
  H(W) = h(\theta)+(1-\theta)H(W\mid W\ne0) \ge h(\theta)+(1-\theta)\log\abs{G_1}.
\]
\end{lemma}
\begin{proof}
Let $E=\1{S=0}$, where note $E=\1{W=0}$ as $0\not\in G_1$.
Thus
\[
  H(W)=H(W,E)=H(E)+H(W\mid E)
  =h(\theta)+(1-\theta)H(W\mid E=0),
\]
and so it remains to show $H(W\mid E=0)\geq\log\abs{G_1}$, which follows from
\begin{align*}
    H(W\mid E=0)
    &\geq H(W\mid S,E=0)
    \\ &= \sum_{s\neq0}\P(S=s\mid E=0)\,H(W\mid S=s)
    \\ &= \sum_{s\neq0}\P(S=s\mid E=0)\,H(sU_{G_1})
    \\ &= \sum_{s\neq0}\P(S=s\mid E=0)\,\log\abs{G_1}
    \\ &= \log\abs{G_1},
\end{align*}
as $U_{G_1}$ and $S$ are independent.
\end{proof}
We now complete the proof of \cref{thm:expander} in the case $\abs{G_1}>\abs{G_2}\abs{G_3}$.
Recall from \cref{def:KR} that $K_R=2+\log C_R\geq2$.
\begin{lemma}\label{lem:regime1-margin}
Define
\[
  \kappa_1 = -\min_{0<\theta\le1/2}
  \bracket{\theta\log\theta-(1-\theta)\log(1-\theta)} = 0.2117313\ldots,
\]
with minimizer $\theta_\star=0.1613782\ldots$ the root of $\theta(1-\theta)=e^{-2}$ in $(0,\frac12)$. If $\abs{G_1}>\abs{G_2}\abs{G_3}$ and $\theta\le\frac12$, then
\[
  h(\theta)+(1-\theta)\log\abs{G_1} \geq \frac12\log\mathcal N -\kappa_1 .
\]
\end{lemma}
\begin{proof}
As $\frac12-\theta\ge0$ from \cref{lem:theta} and $\abs{G_1}>\abs{G_2}\abs{G_3}\ge1$,
\begin{align*}
    h(\theta) + (1-\theta)\log\abs{G_1}-\frac{1}{2}\log\mc N
    &= h(\theta) + \paren{\frac12-\theta}\log\abs{G_1}-\frac12\log\paren{\abs{G_2}\abs{G_3}}
    \\ &\geq h(\theta) - \theta\log\paren{\abs{G_2}\abs{G_3}}.
\end{align*}
If $\theta=0$ this is 0 and we are done.
If $\theta>0$, \cref{lem:theta} gives
$\log(\abs{G_2}\abs{G_3})\le-2\log\theta$, so
\[
  h(\theta) + (1-\theta)\log\abs{G_1}-\frac{1}{2}\log\mc N \ge h(\theta)+2\theta\log\theta
   = \theta\log\theta-(1-\theta)\log(1-\theta).
\]
Minimizing this expression, which we denote $f(\theta)$, over $\theta\in\left(0,\frac{1}{2}\right]$, we compute
\[
  f'(\theta)=\log\paren{\theta(1-\theta)}+2\log e,\qquad
  f''(\theta)=(\log e)\paren{\frac1\theta-\frac1{1-\theta}}>0
  \quad\text{for }0<\theta<\frac12,
\]
so $f$ is strictly convex on $(0,\frac12)$, with $f'(0^+)=-\infty$ and $f'(\frac12)>0$.
Hence $f$ has a unique interior minimum, at $\theta_\star$.
\end{proof}
 
\subsubsection{The case $\abs{G_1}\le\abs{G_2}\abs{G_3}$}
Without loss of generality, we may assume $\abs{G_2}\leq\abs{G_3}$.
Following the proof of
\cite[Lemma 3.2]{GGK}, which in turn extracts the argument from
\cite[Theorem~2.1]{KMPS2020}, define
\begin{align*}
  P&=\set{(a,b',ac)\ :\ (a,b',c)\in G_1\times G_2\times G_3}\subseteq\R^3,\\
  Q&=\set{\rho_{a',b,c'}:=\{(x,y,z)\in\R^3: bx-a'y+z=a'c'\}\ :\
     (a',b,c')\in G_1\times G_2\times G_3}.
\end{align*}
Recall the notation
\[ I(P,Q)=\#\{(x,\rho)\in P\times Q:x\in\rho\} \]
from \cref{thm:incidence}.
We first verify the number of collisions $a(b+c)=a'(b'+c')$ equals $I(P,Q)$.
\begin{lemma}\label{lem:col-inc}
The map $\phi:(a,b',c)\mapsto(a,b',ac)$ is bijective as a map $G_1\times G_2\times G_3\to P$, and $\psi:(a',b,c')\mapsto\rho_{a',b,c'}$ is a well-defined injection of $G_1\times G_2\times G_3$ into the set of planes of $\R^3$.
Consequently $\abs{P}=\abs{Q}=\mathcal N$.
Moreover, letting
\[
  \mathrm{Col} = \#\set{(a,b,c,a',b',c')\in G_1^2\times G_2^2\times G_3^2:\
  a(b+c)=a'(b'+c')},
\]
we have $\mathrm{Col}=I(P,Q)$.
\end{lemma}
\begin{proof}
Bijectivity of $\phi$ is clear as $a\neq0$.
$\rho_{a',b,c'}$ is a plane, i.e., has nonzero normal vector $(b,-a',1)$, and injectivity of $\psi$ follows, after recalling $a'\neq0$.\footnote{\cite[Lemma 3.2]{GGK} does not assume $0\notin G_1$, and thus their subsequent identification $\mathrm{Col}=I(P,Q)$ is erroneous; their conclusion nevertheless survives in the equal-size case they consider, as the tuples with $a=0$ or $a'=0$ contribute $O(\abs A^4)$ to $\mathrm{Col}$ and they eventually bound $I(P,Q)=O(\abs{A}^{9/2})$, which can absorb the $O(\abs A^4)$ error.}

The point $\phi(a,b',c)=(a,b',ac)$ lies on the plane $\psi(a',b,c')$ if and only if
\[
  b\cdot a-a'\cdot b'+ac=a'c'
  \quad\Longleftrightarrow\quad
  a(b+c)=a'(b'+c').
\]
The map $(a,b,c,a',b',c')\mapsto(\phi(a,b',c),\psi(a',b,c'))$ is a bijection from
$G_1^2\times G_2^2\times G_3^2$ onto $P\times Q$, and so the previous equivalence implies $\mathrm{Col}=I(P,Q)$.
\end{proof}
 
We verify the hypotheses of \cref{thm:incidence} with $k=\max\set{\abs{G_1},\abs{G_2}}+1$ and $\ell=1$; note hypotheses (i) and (ii) hold, the latter trivially so.
\begin{lemma}\label{lem:inc-hyp}
Any line $\mc L'\subset\R^2$ satisfies
\[ \abs{\mc L'\cap(G_1\times G_2)}\le\max\set{\abs{G_1},\abs{G_2}}. \]
Thus, if a line $\mc L\subseteq\R^3$ is contained in some plane of $Q$, then
\[ \abs{\mc L\cap P}\le\max\set{\abs{G_1},\abs{G_2}}<k. \]
\end{lemma}
\begin{proof}
If $\mc L'$ is vertical it meets $G_1\times G_2$ in at most $\abs{G_2}$ points; otherwise distinct points of $\mc L'\cap(G_1\times G_2)$ have distinct $x$-coordinates in $G_1$, and so there are at most $\abs{G_1}$ of them.

The plane $\rho_{a',b,c'}$ has normal vector $\mbf n=(b,-a',1)$. If
$\mc L\subseteq\rho_{a',b,c'}$ has direction vector $\mbf d$ then $\mbf d\cdot\mbf n=0$; a vertical $\mbf d=(0,0,\lambda)$ with $\lambda\ne0$ would give $\mbf d\cdot\mathbf n=\lambda\ne0$, contradiction.
Hence $\mc L$ is non-vertical and projects injectively onto a line $\mc L'$ of the plane
$z=0$.
Under the same projection, $P$ projects into $G_1\times G_2$, so
\[ \abs{\mc L\cap P}\le\abs{\mc L'\cap(G_1\times G_2)}\le\max\set{\abs{G_1},\abs{G_2}}<k. \qedhere\]
\end{proof}
Applying \cref{thm:incidence} then completes the proof of \cref{thm:expander}, as follows.
\begin{lemma}\label{lem:col-bound}
If $\abs{G_1}\le\abs{G_2}\abs{G_3}$ and $\abs{G_2}\le\abs{G_3}$, then $\max\set{\abs{G_1},\abs{G_2}}\le\mathcal N^{1/2}$, and
consequently
\[
  \mathrm{Col} \le 4C_R\,\mathcal N^{3/2}.
\]
Thus,
\[ H(W) \geq \frac12\log\mc N - K_R. \]
\end{lemma}
\begin{proof}
If $\max\set{\abs{G_1},\abs{G_2}}=\abs{G_1}$, then
\[ \max\set{\abs{G_1},\abs{G_2}}^2 = \abs{G_1}^2 \leq \abs{G_1}\abs{G_2}\abs{G_3} = \mc N, \]
and otherwise
\[ \max\set{\abs{G_1},\abs{G_2}}^2=\abs{G_2}^2 \leq \abs{G_2}\abs{G_3} \leq \mc N. \]
By \cref{lem:col-inc,lem:inc-hyp}, \cref{thm:incidence}
applies with $|P|=|Q|=\mathcal N$, $\ell=1$ and 
\[ k=\max\set{\abs{G_1},\abs{G_2}}+1\le\mathcal N^{1/2}+1, \]
and gives
\[
  \mathrm{Col}=I(P,Q) \le C_R\paren{\mathcal N^{1/2}\cdot\mathcal N
  +1\cdot\mathcal N+\paren{\mathcal N^{1/2}+1}\mathcal N}
   = C_R\paren{2\mathcal N^{3/2}+2\mathcal N} \leq 4C_R \mc N^{3/2}.
\]
Let $W'=U_{G_1}'(U_{G_2}'+U_{G_3}')$ be an independent copy of $W$ defined via independent copies of $U_{G_1},U_{G_2},U_{G_3}$.
Then
\[ \sum_w\P(W=w)^2=\P(W=W')=\frac{\mathrm{Col}}{\mathcal N^2}, \]
and so
\[ H(W) \geq H_2(W) = 2\log\mc N - \log\mathrm{Col} \geq \frac12\log\mc N - K_R. \qedhere \]
\end{proof}
Having completed both cases $\abs{G_1}>\abs{G_2}\abs{G_3}$ and $\abs{G_1}\leq\abs{G_2}\abs{G_3}$, this completes the proof of \cref{thm:expander}.
\begin{remark}
    The separate argument for the case $\abs{G_1}>\abs{G_2}\abs{G_3}$ is necessary: as a simple example, the incidence argument alone cannot see the degeneracy from \eqref{eq:nondeg}.
    For $G_2=\{b\}$ and $G_3=\{-b\}$ one has $\mathcal N=\abs{G_1}$ while every one of the $\abs{G_1}^2$ tuples is a collision, so $\mathrm{Col}=\mc N^2$, contradicting \cref{lem:col-bound} as $\abs{G_1}=\mc N\to\infty$.
    However, this unbalanced case where $\abs{G_1}$ is large is precisely what the first case covers.
\end{remark}
\begin{remark}\label{rmk:expander-sharp}
The following family of examples shows \cref{thm:expander} is sharp, up to the constant $K_R$ being removed.
Let $G_1\subseteq\R\setminus\{0\}$ be finite and symmetric, let $G_2=\{b_1,b_2\}$, and let $G_3=-G_2$.
Writing $d=b_1-b_2\ne0$, the variable $S=U_{G_2}+U_{G_3}$ is $0$ with probability $\frac12$ and $\pm d$ with probability $\frac14$ each; conditionally on $S\ne0$, $W=U_{G_1}S$ is uniform on $dG_1$ as $G_1$ is symmetric.
So
\[
  H(W)=h\paren{\frac12}+\frac12\log\abs{G_1}=1+\frac12\log\abs{G_1} = \frac12\log\paren{4\abs{G_1}} = \frac12\log\mc N.
\]
This limits the ability of this portion of the argument to improve upon $\delta$, as discussed further in \cref{sec:open}.
\end{remark}
\subsection{Conditioning to uniform}\label{subsec:layered}
In this subsection, we apply the bound from \cref{subsec:expander} to arbitrary laws via our uniformization lemma, \cref{lem:unif-decomp}.
The gain of \cref{lem:unif-decomp} over a direct application of \cref{lem:vadhan} to $X$, as done by Gavalakis, Goh, and Kontoyiannis \cite{GGK}, is that the loss is $\Delta(H)=O(\log H)$ rather than $H-H_\infty(X)$, which can be as large as $H-O(1)$.
This is what allows this section's statements to involve Shannon entropy alone and not min-entropy.
The following is the upgraded analog of \cite[Lemma 3.5]{GGK}.
\begin{theorem}\label{thm:layered-zf}
Let $X,Y,Z$ be independent finitely supported real-valued random variables with positive entropy and $\P(X=0)=0$, and let $\pi_{YZ}=\P(Y+Z=0)$. Then
\[
  H\paren{X(Y+Z)} \ge \frac12\bracket{(1-\pi_{YZ})\paren{H(X)-\Delta(H(X))}
  +H(Y)-\Delta(H(Y))+H(Z)-\Delta(H(Z))}-K_R,
\]
where $\Delta(\eta)=\log(1+\eta)+\log e + 2$.
\end{theorem}
 
\begin{proof}
Apply \cref{lem:unif-decomp} separately to $X$, to $Y$ and to $Z$, obtaining pairs $(X,L)$, $(Y,L')$, $(Z,L'')$ with data $(G_\ell,n_\ell)$, $(G'_{\ell'},n'_{\ell'})$, $(G''_{\ell''},n''_{\ell''})$.
Realize the three couplings on a single product probability space, so that the three labeled
pairs are mutually independent.
Then, conditionally on $(L,L',L'')=(\ell,\ell',\ell'')$, the triple $(X,Y,Z)$ consists of three independent uniform variables on $G_\ell$, $G'_{\ell'}$, $G''_{\ell''}$.
 
As before, let $W=X(Y+Z)$. 
Then
\begin{equation}\label{eq:cond-labels}
  H(W) \ge H(W\mid L,L',L'')=\E\bracket{H\paren{U_{G_L}(U_{G'_{L'}}+U_{G''_{L''}})}}.
\end{equation}
We need to exclude the degenerate case before we can apply \cref{thm:expander}.
Let
\[ \mc E:=\{\exists\,b\in\R:G'_{L'}=\{b\}\text{ and }G''_{L''}=\{-b\}\}, \]
an event that is $(L',L'')$-measurable and independent of $L$.
So $\mc E\subseteq\{Y+Z=0\}$ and thus $\P(\mc E) \le \pi_{YZ}$.
On $\mc E$ both $G'_{L'}$ and $G''_{L''}$ are singletons, so $n'_{L'}=n''_{L''}=0$ on $\mc E$.
 
On $\mc E^c$, the triple $(G_L,G'_{L'},G''_{L''})$ satisfies the hypotheses of \cref{thm:expander}, where
\[ \log\paren{|G_L||G'_{L'}||G''_{L''}|}=n_L+n'_{L'}+n''_{L''}. \]
Hence,
\[ H\paren{U_{G_L}(U_{G'_{L'}}+U_{G''_{L''}})} \ge \frac12\paren{n_L+n'_{L'}+n''_{L''}}-K_R \]
on $\mc E^c$.
As the left-hand side is nonnegative almost surely, including on $\mc E$, \eqref{eq:cond-labels} gives
\begin{align*}
    H(W)
    &\ge \E\bracket{1_{\mc E^c}H\paren{U_{G_L}(U_{G'_{L'}}+U_{G''_{L''}})}}
    \\ &\geq \frac12\bracket{\E\bracket{\1{\mc E^c}n_L}
    +\E\bracket{\1{\mc E^c}n'_{L'}}+\E\bracket{\1{\mc E^c}n''_{L''}}}
    -K_R\,\P(\mc E^c).
\end{align*}
By independence of $\mc E$ from $L$,
\[ \E[\1{\mc E^c}n_L]=\P(\mc E^c)\E[n_L], \]
and as $n'_{L'}=n''_{L''}=0$ on $\mc E$, we have $\E[\1{\mc E^c}n'_{L'}]=\E[n'_{L'}]$ and likewise for $n''$.
Thus, by \cref{lem:unif-decomp},
\begin{align*}
    H(W)
    &\ge \frac12\bracket{\P(\mc E^c)\,\E[n_L]+\E[n'_{L'}]+\E[n''_{L''}]}-K_R
    \\ &\geq \frac12\bracket{(1-\pi_{YZ})\paren{H(X)-\Delta(H(X))}+H(Y)-\Delta(H(Y))+H(Z)-\Delta(H(Z))}-K_R;
\end{align*}
note that $\P(\mc E^c)\E[n_L]\geq(1-\pi_{YZ})(H(X)-\Delta(H(X)))$ is valid even in the case $H(X)<\Delta(H(X))$, as then the left-hand side is nonnegative while the right-hand side is nonpositive.
\end{proof}
\begin{corollary}\label{cor:layered}
Let $\widetilde X,\widetilde X',\widetilde Y,\widetilde Z$ be i.i.d.\ finitely supported real-valued zero-free random variables with $\widetilde H=H(\widetilde X)>0$, and let
$\widetilde\pi=\P(\widetilde X+\widetilde X'=0)$.
Then
\[
    H\paren{\widetilde X(\widetilde Y+\widetilde Z)} \ge
    \frac32\widetilde H-\frac12\widetilde\pi\widetilde H
    -\frac32\log\paren{1+\widetilde H}-K_R-\frac{3}{2}(\log e + 2).
\]
\end{corollary}
 
\section{Large-doubling entropic Elekes--Ruzsa theorem}\label{sec:upper}
As done by Gavalakis, Goh, and Kontoyiannis \cite{GGK}, we now establish an upper bound on the quantity $H(\widetilde X(\widetilde Y+\widetilde Z))$ from \cref{cor:layered}, in terms of sum and product entropies, to combine with our lower bound from \cref{sec:incidence} to get a result of sum-product phenomenon type, specifically the finite support version of our large-doubling linear entropic Elekes--Ruzsa theorem, \cref{thm:baseline}.
As noted by Gavalakis, Goh, and Kontoyiannis \cite{GGK}, a result of this form appeared in the original preprint version of a paper by M\'ath\'e and O'Regan \cite{MO2025}.
However, the proof of the result, by functional submodularity, was incorrect; the result was removed in the published version of that paper.
Thus, Gavalakis, Goh, and Kontoyiannis \cite[Lemma 3.6]{GGK} then provided a proof of this result for completeness; however, their proof is also incorrect, for essentially the same reason as before.
The issue in the proof of M\'ath\'e and O'Regan \cite{MO2025} is that when $X\neq 0$, one can determine $(X,Y,Z)$ from $(X,Y+Z,XY,YZ)$, but this is not true for $X=0$; Gavalakis, Goh, and Kontoyiannis \cite{GGK} correct this issue by conditioning on $X\neq 0$, but then forget to remove this conditioning at the end.
As discussed in \cref{rmk:erratum}, the originally stated bound is simply false.
We prove a corrected upper bound in \cref{prop:upper}, and then in \cref{rmk:erratum} also discuss why using this bound salvages the main results of Gavalakis, Goh, and Kontoyiannis \cite[Theorems 1.3, 1.4, and 1.6]{GGK}.
Note that, with a coincidence in numbering, \cite[Theorems 1.3 and 1.4]{GGK} appear in this paper as \cref{thm:GGK_1_3,thm:GGK_1_4}, respectively.
As in \cref{sec:incidence}, for the sake of rigor our proofs are provided completely.

The analog of this upper bound for continuous real-valued random variables, using differential entropy, was proven by the author, Gavalakis, and Kontoyiannis \cite[Proposition 7.1]{LGK}; as functional submodularity does not hold for differential entropy, its proof is different and immune to this issue (the variables are also continuous, so there are no atoms at 0).
However, the correct expression in this continuous case depends on $\E[\log\abs{X}]$, and so heavily depends on the behavior of $X$ near 0. 

\subsection{The product-sum upper bound}\label{subsec:upper-bound}

\begin{proposition}\label{prop:upper}
Let $X,Y,Z$ be jointly distributed discrete random variables taking values in a field, with
$H(X,Y,Z)<\infty$.
Then
\begin{equation}\label{eq:upper}
H\paren{X(Y+Z)} \le H(XY,XZ)+H(X,Y+Z)-H(X,Y,Z)+\P(X=0)\,H\paren{Y,Z\mid Y+Z,\ X=0}.
\end{equation}
Moreover, the slackness in \eqref{eq:upper} equals exactly
\[ I(R;S)-H(W) \ge 0, \]
where
\[ R=(XY,XZ), \qquad S=(X,Y+Z), \qquad\text{and}\qquad W=X(Y+Z). \]
\end{proposition}
\begin{proof}
Let $T=(X,Y,Z)$.
Both $R,S$ are functions of $T$, so 
\begin{equation}\label{eq:upper-chain}
H(R)+H(S)=H(R,S)+I(R;S)=H(T)-H(T\mid R,S)+I(R;S).
\end{equation}
Next $W=f(R)$ with $f(r_1,r_2)=r_1+r_2$ and $W=g(S)$ with
$g(s_1,s_2)=s_1s_2$, so by the data processing property of mutual information, i.e., (E3) of \cref{lem:entropy-facts}, applied twice,
\begin{equation}\label{eq:upper-dp}
I(R;S) \ge I(f(R);g(S))=I(W;W)=H(W).
\end{equation}
Finally we compute $H(T\mid R,S)$ exactly.
The event $\{X=0\}$ is $S$-measurable.
On $\{X\ne0\}$, the pair $(R,S)$ determines $T$, so there is no conditional entropy in $T$, while on $\{X=0\}$ one has $R=(0,0)$ and $S=(0,Y+Z)$, so conditioning on $(R,S)$ is conditioning on $Y+Z$, while $X=0$ is determined.
Therefore
\[ H(T\mid R,S)=\P(X=0)\,H\paren{T\mid R,S,\ X=0}=\P(X=0)\,H\paren{Y,Z\mid Y+Z,\ X=0}. \]
Substituting this and \eqref{eq:upper-dp} into \eqref{eq:upper-chain} gives \eqref{eq:upper},
with slackness exactly $I(R;S)-H(W)\ge0$.
\end{proof}
\begin{corollary}\label{cor:upper-zf}
Let $X,Y,Z$ be i.i.d.\ finitely supported random variables with values in a field, let $H=H(X)$, $p_0=\P(X=0)$, and let $X'$ be an independent copy of $X$. Then
\begin{equation}\label{eq:upper-iid}
\begin{split}
H\paren{X(Y+Z)} &\le 2H(XX')+H(X+X')-2H+p_0\paren{2H-H(X+X')}\\
&= 2H(XX')-(1-p_0)\paren{2H-H(X+X')}.
\end{split}
\end{equation}
In particular, for a zero-free law $\widetilde X$, with i.i.d.\ copies $\widetilde Y,\widetilde Z,\widetilde X'$ and entropy $\widetilde H=H(\widetilde X)$,
\begin{equation}\label{eq:upper-zf}
H\paren{\widetilde X(\widetilde Y+\widetilde Z)} \le
2H(\widetilde X\widetilde X')+H(\widetilde X+\widetilde X')-2\widetilde H .
\end{equation}
\end{corollary}
\begin{remark}[Erratum]\label{rmk:erratum}
The statement corrected by \cref{prop:upper} entered the literature as \cite[Lemma 3.6]{GGK},
which asserts, for arbitrary field-valued $X,Y,Z$ and then in the i.i.d.\ case,
\begin{align}
    H\paren{X(Y+Z)} &\le H(XY,XZ)+H(X,Y+Z)-H(X,Y,Z)+1, \label{eq:ggk15}\\
    H\paren{X(Y+Z)} &\le 2H(XY)+H(X+Y)-2H(X)+1. \label{eq:ggk16}
\end{align}
Both are false without a correction term: \eqref{eq:ggk15} can fail by an unbounded margin for every $\P(X=0)>0$, and \eqref{eq:ggk16} is false whenever $\P(X=0)>\frac12$; see
\cref{ex:15,ex:16}.
The difficulty lies in moving from the last display in the proof of \cite[Lemma 3.6]{GGK} to \eqref{eq:ggk15}, where they stop conditioning on the indicator $\1{X=0}$.
\Cref{prop:upper} repairs the argument by retaining
$\P(X=0)H(Y,Z\mid Y+Z,X=0)$ exactly; when $p_0=0$ it is strictly stronger than
\eqref{eq:ggk15}--\eqref{eq:ggk16}, the additive $+1$ being unnecessary.
 
At the level of their main theorems, their Theorems 1.3 and 1.6 survive exactly as stated.
Theorem 1.3 combines their Lemma 3.5 with \eqref{eq:ggk16}; with \eqref{eq:upper-iid} in
its place the same argument yields
\begin{equation}\label{eq:ggk-corrected}
    \max\set{H(X+X'),H(XX')}\ge\frac16(4H+3H_\infty(X)-4p_0H)-O(1),
\end{equation}
where $H=H(X)$, and (i) if $H_\infty(X)\le\frac23H$ their conclusion is trivial, since
\[ \frac16(4H+3H_\infty(X))\le H\le H(X+X'), \]
while (ii) if $H_\infty(X)>\frac23H$ then
$p_0\le p_{\max}=\max_x \P(X=x)<2^{-2H/3}$, so
\[ p_0H\le H2^{-2H/3}\le\frac3{2e\ln2}<0.797 \]
can be absorbed by the $O(1)$.
Theorem 1.6 uses Theorem 1.3 only through its statement, so is also unaffected.
Theorem 1.4 is proved instead from \cite[Lemmas 4.2 and 4.3]{GGK}, and the latter has the same un-conditioning mistake with exceptional event $E^c=\{X+Y=0\text{ or }Z+W=0\}$, whose probability is not killed by the assumption $\P(X=0)=0$.
The analog of \cref{prop:upper} for it would carry the correction term $\P(E^{c})H(X,Y,Z,W\mid X+Y,Z+W,E^{c})$, and then for the i.i.d.\ case, which is what they need for Theorem 1.4, we can split into the cases of $E^c$ where $X+Y=0$ or $Z+W=0$, and in each case the analog of the $p_0H$ term from \eqref{eq:ggk-corrected} is
\[ \P(X+Y=0)H \leq p_{\max} H = O(1) \]
for $H_\infty(X)>\frac12H(X)$ by an identical argument, the only regime in which Theorem 1.4 has content.
Here, we use the fact that
\[ \P(X+Y=0) =\sum_x \P(X=x)\P(Y=-x) \leq p_{\max} \sum_x \P(X=x) = p_{\max}. \]
\end{remark}
 
\begin{example}[A counterexample to \eqref{eq:ggk15}]\label{ex:15}
Let $p_0\in(0,1)$, let $X$ take the values $0$ and $1$ with probabilities $p_0$ and $1-p_0$, respectively, and let $Y,Z$ be independent of $X$ and of each other and uniform on $\{1,2,\dots,n\}$.
Then
\[
H\paren{X(Y+Z)}-\bracket{H(XY,XZ)+H(X,Y+Z)-H(X,Y,Z)+1}
=p_0\paren{2\log n-H(Y+Z)}-1 ,
\]
which tends to $+\infty$ as $n\to\infty$.
\end{example}
\begin{proof}
Write $q=1-p_0$.
Since $Y+Z\ge2>0$, the indicator $\1{X=1}=X$ is a function of $X(Y+Z)$, so 
\[ H(X(Y+Z)) = H(X) + H(X(Y+Z)|X) = h(p_0) + qH(Y+Z). \]
Similarly, $X$ is a function of $(XY,XZ)$, so
\[ H(XY,XZ)= H(X) + H(XY,XZ|X) = h(p_0)+2q\log n. \]
By independence,
\[ H(X,Y+Z)=h(p_0)+H(Y+Z) \qquad\text{and}\qquad H(X,Y,Z)=h(p_0)+2\log n. \]
Hence,
\[ H\paren{X(Y+Z)}-\bracket{H(XY,XZ)+H(X,Y+Z)-H(X,Y,Z)+1} = p_0(2\log n-H(Y+Z))-1, \]
as desired.
Since $H(Y+Z)\le\log(2n-1)=\log n+O(1)$, the right-hand side is $p_0\log n+O(1)\to\infty$.
\end{proof}
The i.i.d.\ form \eqref{eq:ggk16} is more robust, the family above using variables that are not identically distributed, but it too fails for every $p_0>\frac12$.
\begin{example}[A counterexample to \eqref{eq:ggk16}]\label{ex:16}
Let $n$ be a positive integer, set $N=2n^2$ and $A=\{N+1,N+2,\dots,N+n\}$. Let $p_0\in(0,1)$, with $X$ equal to $0$ with probability $p_0$ and otherwise uniform on $A$, and let $Y,Z,X'$ be independent copies of $X$. Then
\begin{equation}\label{eq:counterex16}
    H\paren{X(Y+Z)}-\bracket{2H(XX')+H(X+X')-2H(X)+1}=(1-p_0)^2(2p_0-1)\log n+O(1),
\end{equation}
which tends to $+\infty$ as $n\to\infty$ for any $p_0>\frac12$.
\end{example}
\begin{proof}
    The calculations here are routine but tedious, so we will be more brief than usual.
    Note that $N$ was chosen sufficiently large so that for $i,j,i',j'\in\set{1,2,\dots,n}$,
    \[ (N+i)(N+j) = (N+i')(N+j') \iff (i,j)=(i',j') \text{ or }(i,j)=(j',i'): \]
    this is because $1 \leq ij \leq n^2 < N$, so the values of $ij$ and $i+j$ can be read off from the value of $(N+i)(N+j)=N^2+N(i+j)+ij$, which determines $i$ and $j$ up to ordering.
    Then conditional on $\set{X\neq0,X'\neq0}$, the product $XX'$ has the same distribution, up to a bijective relabeling, as an unordered pair from $\set{1,2,\dots,n}^2$, so it is straightforward to calculate
    \[ H(XX') \geq H(XX'|X>0,X'>0) = 2(1-p_0)^2\log n + O(1), \]
    and similarly
    \[ H(XX') \leq H(XX',\1{X>0,X'>0}) = H(XX'|\1{X>0,X'>0}) + H(\1{X>0,X'>0}) = 2(1-p)^2\log n + O(1). \]
    
    A similar argument, conditioning on the $O(1)$ bits of information of whether $X$ or $X'$ are 0, yields
    \[ H(X+X')=p_0^2\cdot0+2p_0(1-p_0)\log n + (1-p_0)^2\log n + O(1) = (1-p_0^2)\log n + O(1); \]
    this uses the computation that $H(X+X'\mid X>0,X'>0)$, which is a triangular distribution on $2n-1$ values, has entropy $\log n + O(1)$, which can be verified as
    \[ \log n = H(X\mid X>0) \leq H(X+X'\mid X>0, X'>0) \leq \log(2n-1) \leq \log n + 1, \]
    by independence of $X$ and $X'$, so the first inequality follows by conditioning on $X'$.
    
    We have $\P(X(Y+Z)\neq0)=(1-p_0)(1-p_0^2)$.
    Conditional on $\set{X(Y+Z)\neq 0}$, similar to our computation of $H(XX')$, note that for $X=N+i$, $Y=N+j$, and $Z=N+k$, we have
    \[ X(Y+Z) = 2N^2 + N(2i+j+k) + i(j+k), \]
    with $2 \leq i(j+k) \leq 2n^2 \leq N$, so $X(Y+Z)$ determines $2i+j+k$ and $i(j+k)$, which up to ordering determines the pair $(i,j+k)$ (in some cases $j+k>n$, in which case the ordering is determined as $i\leq n$), so similar to the $H(XX')$ argument, this ordering impacts the entropy by at most a bit, and so
    \[ H(X(Y+Z)\mid X(Y+Z)\neq 0) = H(i,j+k\mid X(Y+Z)\neq 0) + O(1) = 2 \log n + O(1), \]
    where the last equality follows as $i$ and $j+k$ are conditionally independent, with $j+k$ essentially being a triangular distribution so a similar argument as for $H(X+X')$ yields it has entropy $\log n + O(1)$.
    Thus,
    \[ H(X(Y+Z)) = 2(1-p_0)(1-p_0^2)\log n + O(1). \]

    Lastly, $H(X) = (1-p_0)\log n + O(1)$, so combining our computations yields \eqref{eq:counterex16}.
\end{proof}
\subsection{Assembling together}\label{subsec:baseline}
In this subsection we prove \cref{prop:baseline}, the finite support version of \cref{thm:baseline}.

Recall the notation
\[ p_0=\P(X=0), \qquad p_{\max}=\max_x\P(X=x), \qquad\text{and}\qquad \pi=\P(X+X'=0). \]
Also recall from \cref{rmk:erratum} that $\pi\leq p_{\max}$.
We first record the following lemma, which we will need for \cref{prop:baseline}.
It is a slight improvement on \cite[Proposition 5.1]{GGK}, which carries the additive loss $\frac32$ in place of $h(p_{\max})\leq1$, and has a simpler, self-contained proof, where the proof of \cite[Proposition 5.1]{GGK} also needs a computation from \cite{GMT2025}.
\begin{lemma}\label{lemma:maxatom}
Let $X$ take values in an abelian group with $H(X)<\infty$ and let $X'$ be an independent copy of $X$. Then
\begin{equation}\label{eq:maxatom-used}
p_{\max}H(X) \le H(X+X')-H(X)+h(p_{\max}) \le H(X+X')-H(X)+1 .
\end{equation}
\end{lemma}
\begin{proof}
The second inequality is trivial as $h(\cdot)\leq1$.
Let $a$ satisfy $p_{\max}=\P(X=a)$.
Then
\[ H(X+X') \ge H\paren{X+X'\mid\1{X'=a}}=p_{\max}\,H(X+a)+(1-p_{\max})\,H(X+X^*), \]
where $X^*\sim\mathrm{Law}(X'\mid X'\ne a)$ is independent of $X$.
As $H(X+a)=H(X)$,
\[ H(X+X^*)\ge H(X+X^*\mid X)=H(X^*), \]
and
\[ H(X) = H(X,\1{X=a}) = h(p_{\max}) + H(X|\1{X=a}) = h(p_{\max}) + (1-p_{\max})H(X^*), \]
we have
\[ H(X+X')\ge p_{\max}H(X)+(1-p_{\max})H(X^*) \geq p_{\max}H(X) + H(X)-h(p_{\max}), \]
which rearranges to the result.
\end{proof}
Finally, we combine \cref{cor:layered} and \cref{cor:upper-zf}, to prove the main result of this section.
Recall from \cref{thm:layered-zf} the per-variable uniformization loss, which we denote by 
\[ \Delta(\eta)=\log(1+\eta)+\log e+2 \]
and write
\begin{equation}\label{eq:E1}
    E_1(\eta) = \frac32\Delta(\eta)+K_R
\end{equation}
for the error term of \cref{cor:layered}, whose conclusion then reads
\[ H(\widetilde X(\widetilde Y+\widetilde Z))\ge\frac32\widetilde H-\frac12\widetilde\pi\widetilde H-E_1(\widetilde H). \]
\begin{proposition}\label{prop:baseline}
Let $\widetilde X$ be zero-free and finitely supported with $\widetilde H=H(\widetilde X)>0$. Then, using the notation of $\widetilde u$ and $\widetilde m$ from \cref{thm:eer},
\begin{equation}\label{eq:baseline}
\widetilde m \ge F_1(\widetilde u)-\vartheta(\widetilde H), \qquad F_1(\widetilde u) = \frac54-\frac34\widetilde u,
\end{equation}
where
\begin{equation}\label{eq:theta1}
\vartheta(\eta) = \frac{E_1(\eta)+\frac12}{2\eta}
 = \frac{\frac34\log(1+\eta)+\frac12K_R+\frac34(\log e+2)+\frac14}{\eta}.
\end{equation}
$\vartheta:(0,\infty)\to(0,\infty)$ is a nonincreasing function.
\end{proposition}
 
\begin{proof}
By \cref{cor:layered} and the zero-free upper bound \eqref{eq:upper-zf} of
\cref{cor:upper-zf},
\[
\frac32\widetilde H-\frac12\widetilde\pi\widetilde H-E_1(\widetilde H)
 \le H\paren{\widetilde X(\widetilde Y+\widetilde Z)}
 \le 2\widetilde m\widetilde H+(1+\widetilde u)\widetilde H-2\widetilde H.
\]
Rearranging,
\begin{equation}\label{eq:baseline-mid}
2\widetilde m\widetilde H \ge
\paren{\frac52-\widetilde u-\frac12\widetilde\pi}\widetilde H-E_1(\widetilde H).
\end{equation}
Define
\[ \widetilde\pi = \P(\widetilde X + \widetilde X'=0) \qquad\text{and}\qquad \widetilde p_{\max} = \max_x \P(\widetilde X=x), \]
where $\widetilde X'$ is an independent copy of $\widetilde X$.
From \cref{rmk:erratum} we have $\widetilde \pi \leq \widetilde p_{\max}$.
Then \cref{lemma:maxatom} gives
\[ \widetilde\pi\le\widetilde p_{\max}\le\widetilde u+1/\widetilde H; \]
hence \eqref{eq:baseline-mid} becomes
\[
2\widetilde m\widetilde H \ge \paren{\frac52-\frac32\widetilde u}\widetilde H
-E_1(\widetilde H)-\frac12,
\]
which dividing by $2\widetilde H$ gives the result.
$\vartheta$ is clearly positive, and it is straightforward to check that it is nonincreasing.
\end{proof}
\section{Proof of the main theorem}\label{sec:main}
We now have two lower bounds, \cref{cor:eer-eff,prop:baseline}, for the relative product entropy $\widetilde m$ of a zero-free law, both linear in $\widetilde u$ and both valid with no hypothesis beyond being zero-free and finitely supported.
We denote these linear bounds via the lines
\[
F_1(\widetilde u)=\frac54-\frac34\widetilde u \quad\text{(\cref{prop:baseline})},
\qquad
F_2(\widetilde u)=2-6\widetilde u \quad\text{(\cref{cor:eer-eff})}.
\]
We first reincorporate atoms at zero, proving the finite support version of \cref{thm:main} as \cref{thm:main-body}, and then provide a truncation argument in \cref{lem:truncation} that extends this result to the general case; we also extend \cref{cor:eer-eff,prop:baseline} to yield \cref{thm:eer,thm:baseline}, respectively.

Recall our notation from \cref{sec:prelim}: $H=H(X)$ with zero-free law $\widetilde X\sim\operatorname{Law}(X|X\neq0)$ of entropy $\widetilde H=H(\widetilde X)$.
And the four expressions
\[
  u=\frac{H(X+X')}{H}-1,\qquad m=\frac{H(XX')}{H},\qquad
  \widetilde u=\frac{H(\widetilde X+\widetilde X')}{\widetilde H}-1,\qquad
  \widetilde m=\frac{H(\widetilde X\widetilde X')}{\widetilde H},
\]
where $X'$ is an independent copy of $X$ and $\widetilde X'$ is an independent copy of $\widetilde X$.
Let $p_0=\P(X=0)$.
We first lower bound $u$ in terms of $\widetilde u$ and $p_0$.

\begin{lemma}\label{lem:transfer}
Let $X$ be a discrete, finitely supported, real-valued random variable with $0\le p_0<1$, and let $X'$ be an independent copy.
Suppose $H=H(X)>1$. Then
\[ u \geq p_0 + (1-p_0)\widetilde u - \frac{2}{H-1}. \]
\end{lemma}
\begin{proof}
We have
\begin{equation}\label{eq:transfer-1}
    H = H(X,\1{X=0}) = h(p_0) + H(X\mid\1{X=0}) = h(p_0) + (1-p_0)\widetilde H.
\end{equation}
Similarly,
\begin{align}
    H(XX')
    &= H(XX',\1{XX'=0}) \label{eq:transfer-2}
    \\ &= h((1-p_0)^2) + H(XX'\mid\1{XX'=0}) \notag
    \\ &= h((1-p_0)^2)+(1-p_0)^2H(\widetilde X\widetilde X') \notag
    \\ &= h((1-p_0)^2) + (1-p_0)^2\widetilde m\widetilde H. \notag
\end{align}
We bound
\begin{align*}
    H(X+X')
    &\geq H(X+X'\mid\1{X=0},\1{X'=0})
    \\ &= 2p_0(1-p_0)\widetilde H + (1-p_0)^2H(\widetilde X + \widetilde X')
    \\ &= (1-p_0)\widetilde H\bracket{2p_0 + (1-p_0)(1+\widetilde u)}
    \\ &= (1-p_0)\widetilde H\bracket{1 + p_0 + (1-p_0)\widetilde u}
    \\ &= (H-h(p_0))\bracket{1 + p_0 + (1-p_0)\widetilde u},
\end{align*}
where the last equality uses \eqref{eq:transfer-1}.
This rearranges to
\[ 1 + p_0 + (1-p_0)\widetilde u \leq \frac{(1+u)H}{H-h(p_0)} = (1+u)\paren{1+\frac{h(p_0)}{H-h(p_0)}}, \]
or
\[ p_0 + (1-p_0)\widetilde u \leq u + \frac{(1+u)h(p_0)}{H-h(p_0)} \leq u + \frac{2}{H-1}, \]
which rearranges to the result, recalling $u\leq 1$ from (E5') of \cref{lem:entropy-facts}.
\end{proof}
In the hope that future work can provide better linear entropic Elekes--Ruzsa bounds, we state in larger generality how such zero-free bounds translate to values of $\delta$ for the entropic sum-product phenomenon.
\begin{proposition}\label{prop:master}
Let $F\colon[0,1)\to\R$, let $\vartheta\colon(0,\infty)\to[0,\infty)$ be nonincreasing with
$\vartheta(\eta)\to0$ as $\eta\to\infty$, and suppose that every zero-free finitely supported
real-valued $\widetilde X$ with $\widetilde H>0$ satisfies
\begin{equation}\label{eq:master-hyp}
\widetilde m \ge F(\widetilde u)-\vartheta(\widetilde H).
\end{equation}
Let
\[
\Phi = \inf_{0\le\widetilde u<1}\frac{F(\widetilde u)}{1-\widetilde u},
\]
assume $\Phi>1$, and let $\delta\in\left(0,\frac{\Phi-1}{\Phi+1}\right]$.
Then for all i.i.d.\ finitely supported real-valued $X,X'$ of entropy
\begin{equation}\label{eq:master-thresh}
H > 1+\frac2{1-\delta},
\end{equation}
one has
\begin{equation}\label{eq:master}
\max\set{H(X+X'),H(XX')} \ge (1+\delta)H
-\paren{\frac{2\Phi}{H-1}+\vartheta(H-1)+\frac2H}H .
\end{equation}
In particular, if in addition $F\le2$ and $\vartheta(\eta)=\frac{a\log \eta+b}{\eta}$ for all $\eta\ge H-1$, with constants $a,b\ge0$, then 
\begin{equation}\label{eq:master-sharp}
\max\set{H(X+X'),H(XX')} \ge (1+\delta)H-\frac{2\Phi H}{H-1}-a\log H-b-2 .
\end{equation}
\end{proposition}
\begin{proof}
If $u\ge\delta$ then $H(X+X')=(1+u)H\ge(1+\delta)H$ and we are done.
So assume $u<\delta$ and set $\alpha=\frac2{H-1}$,
so that the hypothesis \eqref{eq:master-thresh} says exactly $\alpha<1-\delta$, so
\begin{equation*}
u+\alpha<\delta+\alpha<1 .
\end{equation*}
Since $H>1$ we have $p_0<1$, so \cref{lem:transfer} applies: for brevity, let $q=1-p_0$, so that
\begin{equation*}
p_0+q\widetilde u \le u+\alpha < 1 .
\end{equation*}
In particular $\widetilde u<1$, so $F(\widetilde u)$ is defined.
This rearranges to
\begin{equation}\label{eq:master-q}
q = 1-p_0 \ge \frac{1-u-\alpha}{1-\widetilde u} > 0 .
\end{equation}
By \eqref{eq:transfer-1},
\begin{equation}\label{eq:master-L}
    \widetilde H = \frac{H-h(p_0)}{q} \geq H-1 > 0,
\end{equation}
so by \eqref{eq:transfer-2},
\begin{equation}\label{eq:master-prod}
    H(XX') \geq q^2\widetilde m\widetilde H = q\widetilde m(H-h(p_0)) \geq q\widetilde m H - 2,
\end{equation}
recalling $\widetilde m \leq 2$ as $H(\widetilde X\widetilde X') \leq H(\widetilde X,\widetilde X') = 2\widetilde H$.
These two inequalities, along with $\vartheta$ being nonincreasing, yield
\[
    m \ge q\widetilde m-\frac2H
    \ge q\paren{F(\widetilde u)-\vartheta(\widetilde H)}-\frac2H
    \ge qF(\widetilde u)-\vartheta(H-1)-\frac2H .
\]
Thus, by \eqref{eq:master-q}
\begin{equation}\label{eq:master-qF}
qF(\widetilde u) \ge \frac{1-u-\alpha}{1-\widetilde u}\,F(\widetilde u) \ge \Phi(1-u-\alpha).
\end{equation}
Therefore, using $u<\delta$ and rearranging the assumption $\delta\leq\frac{\Phi-1}{\Phi+1}$ to yield $\Phi(1-\delta)\ge1+\delta$,
\begin{equation}\label{eq:master-3}
m \ge \Phi(1-u-\alpha)-\vartheta(H-1)-\frac2H
 \ge \Phi(1-\delta)-\Phi\alpha-\vartheta(H-1)-\frac2H
 \ge 1+\delta-\frac{2\Phi}{H-1}-\vartheta(H-1)-\frac2H .
\end{equation}
Multiplying by $H$ gives \eqref{eq:master}.
 
Now assume $F\le2$ and $\vartheta(\eta)=\frac{N(\eta)}{\eta}$ for $\eta\ge H-1$, where
$N(\eta)=a\log\eta+b$ with $a,b\ge0$.
By \eqref{eq:master-prod},
\begin{align*}
    H(XX')
    &\geq q\widetilde m(H-h(p_0))
    \geq q\paren{F(\widetilde u)-\vartheta(\widetilde H)}(H-h(p_0))
    \geq qHF(\widetilde u)-2-q^{2}N(\widetilde H),
\end{align*}
where the last step used $q\widetilde H=H-h(p_0)$ from \eqref{eq:master-L} when multiplying with $\vartheta(\widetilde H)$, and $F\leq2$.
We claim that
\begin{equation}\label{eq:master-N}
q^{2}N(\widetilde H) \le N(H).
\end{equation}
Indeed $\widetilde H\le H/q$ by \eqref{eq:master-L}, so
\[ q^{2}N(\widetilde H)\le q^{2}N(H)+aq^{2}\log(1/q), \]
and with $v=q^{2}$ and
$\ln(1/v)\le\frac{1-v}{v}$,
\[
aq^{2}\log\frac1q=\frac a2\,v\log\frac1v\le\frac{a\log e}2\,(1-v)\le a(1-q^{2})
 \le (1-q^{2})N(H),
\]
because $a\le a\log H\le N(H)$, the first inequality holding since
$H>1+2/(1-\delta)\ge3$ forces $\log H>1$.
Adding the two estimates gives \eqref{eq:master-N}.
Consequently, by \eqref{eq:master-qF} and the same bounding from \eqref{eq:master-3},
\[
    m \ge qF(\widetilde u)-\frac{N(H)+2}H \ge \Phi(1-u-\alpha)-\frac{N(H)+2}H
    \ge 1+\delta-\frac{2\Phi}{H-1}-\frac{a\log H+b+2}H,
\]
and multiplying by $H$ gives \eqref{eq:master-sharp}, recalling in the case $u\geq\delta$ the inequality trivially holds.
\end{proof}
Thus equipped, we prove the finite support version of our main theorem, \cref{thm:main}.
\begin{theorem}\label{thm:main-body}
There are absolute constants $c_1,c_2,H_0$ such that for all i.i.d.\ finitely supported
real-valued $X,X'$ with $H=H(X)\ge H_0$,
\begin{equation}\label{eq:main-eff}
\max\set{H(X+X'),\,H(XX')} \ge \frac87H-c_1\log H-c_2-\frac12\log C_R .
\end{equation}
One may take
\[ c_1=18,\qquad c_2=63,\qquad\text{and}\qquad H_0=5. \]
\end{theorem}
 
\begin{proof}
For an affine function $A-C\widetilde u$, for $\widetilde u\in[0,1)$ we have
\[ \frac{d}{d\widetilde u}\frac{A-C\widetilde u}{1-\widetilde u} = \frac{A-C}{(1-\widetilde u)^2}, \]
so $\frac{A-C\widetilde u}{1-\widetilde u}$ is strictly increasing if $A>C$ and strictly decreasing if $A<C$.
Thus, for our functions
\[
F_1(\widetilde u)=\frac54-\frac34\widetilde u,
\qquad\text{and}\qquad
F_2(\widetilde u)=2-6\widetilde u,
\]
we have $\frac{F_1(\widetilde u)}{1-\widetilde u}$ is strictly increasing and $\frac{F_2(\widetilde u)}{1-\widetilde u}$ is strictly decreasing, and they intersect when $F_1(\widetilde u)=F_2(\widetilde u)$ at $\widetilde u=\frac{1}{7}$.
Thus, letting $F=\max\set{F_1,F_2}$, where $F\leq 2$, we have
\[ \Phi = \inf_{0\leq \widetilde u < 1 }\frac{F(\widetilde u)}{1-\widetilde u}=\frac{F(1/7)}{1-1/7}=\frac43, \]
and so $\frac{\Phi-1}{\Phi+1}=\frac17$.

Recall the error functions associated to $F_1$ and $F_2$ from \cref{prop:baseline} and \cref{cor:eer-eff}, respectively, are
\[
    \vartheta_1(\eta)
    =\frac{\frac34\log(1+\eta)+\frac12K_R+\frac34(\log e+2)+\frac14}{\eta}
    \qquad\text{and}\qquad
    \vartheta_2(\eta)
    =\begin{cases}\frac{18\log\eta+57}{\eta}&\text{if }\eta\ge4, \\ 
    \vartheta_2(4)=\frac{93}4 & \text{if }0<\eta<4, \end{cases}
\]
respectively.
Define
\[ N(\eta) = 18\log\eta+57+\frac12\log C_R \qquad\text{and}\qquad \vartheta(\eta) = \begin{cases} \frac{N(\eta)}{\eta} & \text{if }\eta\geq4, \\ \frac{N(4)}{4} & \text{if }0<\eta<4.\end{cases} \]
Then $\vartheta\geq\max\set{\vartheta_1,\vartheta_2}$ on $[4,\infty)$, and is nonnegative, tends to 0, and is nonincreasing.
It thus satisfies the conditions of \cref{prop:master}, except for potentially \eqref{eq:master-hyp} for $\widetilde H < 4$, but in this case the condition is trivial as $\widetilde m \geq 0 \geq F(\widetilde u)-\vartheta(\widetilde H)$ because $F\leq2$ and $\vartheta(\widetilde H)\geq\frac{93}{4}$.

As $\frac{\Phi-1}{\Phi+1}=\frac17$, \cref{prop:master} holds with $\delta=\frac17$.
The threshold \eqref{eq:master-thresh} reads $H>\frac{10}{3}$ and is implied by $H\geq H_0=5$.
Since $H\geq5$ gives $H-1\geq4$, we have $\vartheta(\eta)=\frac{18\log\eta + 57+\frac12\log C_R}{\eta}$ for all $\eta\geq H-1$, so \eqref{eq:master-sharp} gives
\[
\max\set{H(X+X'),H(XX')} \ge \frac87H-\frac{2\Phi H}{H-1}-18\log H-59-\frac12\log C_R.
\]
For $H\ge5$ we have $H/(H-1)\le\frac54$, so $2\Phi H/(H-1)\le\frac83\cdot\frac54=\frac{10}3\le4$, which yields \eqref{eq:main-eff}.
\end{proof}
\begin{remark}\label{rmk:1_9}
    Note that with just $F_1$, one would have $\Phi=F_1(0)=\frac54$ and thus $\delta=\frac{\Phi-1}{\Phi+1}=\frac19$.
    However, as $F_2(\widetilde u)<0$ for $\widetilde u>\frac13$, with just $F_2$ one would have $\Phi=-\infty$, and would not be able to prove any entropic sum-product phenomenon allowing for atoms at 0.
\end{remark}
\subsection{Removing the finite-support hypothesis}
The finite support hypothesis may be discarded at the outermost level, without revisiting the proofs, via the following simple bound on the entropies of the variable, its i.i.d.\ sum, and its i.i.d.\ product upon truncating a random variable to a finite support.
\begin{lemma}\label{lem:truncation}
Let $X$ be a discrete real-valued random variable with mass function $p$ and
entropy $H=H(X)<\infty$, and let $\theta\in\left(0,\frac12\right]$.
Then there is a finite $F\subseteq\operatorname{supp}X$ with
\[
\P(X\notin F) \le \theta,\qquad
T_F:=\sum_{x\notin F}p(x)\log\frac1{p(x)} \le \theta ,
\]
and for every such $F$ the finitely supported variable
$X_F=\mathrm{Law}(X\mid X\in F)$, with independent copy $X_F'$, satisfies
\begin{enumerate}[label=(\roman*)]
    \item $\displaystyle H-3\theta \le H(X_F) \le \frac H{1-\theta}$;
    \item $\displaystyle H(X_F+X_F')\le\frac{H(X+X')}{(1-\theta)^{2}}$ and
    $\displaystyle H(X_FX_F')\le\frac{H(XX')}{(1-\theta)^{2}}$.
\end{enumerate}
\end{lemma}
 
\begin{proof}
The series $\sum_xp(x)=1$ and $\sum_xp(x)\log(1/p(x))=H$ both converge, so along any exhaustion of $\operatorname{supp}X$ by finite sets their tails tend to $0$; fix $F$ with both at most $\theta$, and put $P_F=\P(X\in F)\ge1-\theta>0$.
The mass function of $X_F$ being $p/P_F$ on $F$, we have
\[ H(X_F)=P_F^{-1}\sum_{x\in F}p(x)\paren{\log(1/p(x))+\log P_F}
=\frac{H-T_F}{P_F}+\log P_F. \]
Thus
\[ H(X_F) \leq \frac{H}{P_F} \leq \frac{H}{1-\theta}, \]
and
\[ H(X_F) \geq H-T_F + \log P_F \geq H - \theta + \log(1-\theta) \geq H - 3\theta, \]
using the inequality $\log(1-\theta)\geq-2\theta$ for $\theta\in\left(0,\frac12\right]$.

Lastly, we bound
\[
  H(X+X')\ge H\paren{X+X'\mid\1{X\in F},\1{X'\in F}}
  \ge P_F^{2}H\paren{X+X'\mid X\in F,X'\in F}=P_F^{2}H(X_F+X_F'),
\]
and then recall $P_F\geq1-\theta$.
The bound for the product is identical.
\end{proof}
\cref{thm:main} with its explicit constants follows from the following corollary.
\begin{corollary}\label{cor:main-general}
    Let $X,X'$ be i.i.d.\ discrete real-valued random variables with $6\leq H=H(X)<\infty$. Then
    \begin{equation}\label{eq:main-general}
    \max\set{H(X+X'),\,H(XX')} \ge \frac87H-18\log H-65-\frac12\log C_R .
    \end{equation}
\end{corollary}
 
\begin{proof}
Apply \cref{lem:truncation} with $\theta=H^{-2}\le\frac1{36}\le\frac12$.
By (i),
\[ H(X_F) \ge H-3H^{-2} \ge6-\frac1{12} > 5, \] 
so $X_F$ being finitely supported, \cref{thm:main-body} and then (i) gives
\begin{align*}
    \max\set{H(X_F+X_F'),H(X_FX_F')}
    &\geq \frac87 H(X_F) -18\log H(X_F)-63-\frac12\log C_R
    \\ &\geq \frac87 H - \frac{24}{7}\theta - 18\log H + 18\log(1-\theta) - 63-\frac12\log C_R
    \\ &\geq \frac87 H - 18\log H - 64 - \frac{2}{21} - \frac12\log C_R,
\end{align*}
where we again use $\log(1-\theta)\geq-2\theta$ for $\theta\in\left(0,\frac12\right]$.
By (ii), where for brevity let
\[ B=\max\set{H(X_F+X_F'),H(X_FX_F')}, \]
we have
\begin{align*}
    \max\set{H(X+X'),H(XX')}
    &\geq (1-\theta)^2 B \geq B - 2\theta B \geq B - \frac{4\theta H}{1-\theta} = B-\frac{4}{H(1-\theta)} \geq B - \frac{24}{35},
\end{align*}
where we use the trivial bound $B\leq 2H(X_F) \leq \frac{2H}{1-\theta}$.
Combining this with our lower bound on $B$ completes the proof.
\end{proof}
 
The same truncation removes the finite support condition from \cref{cor:eer-eff}, proving \cref{thm:eer}.
\begin{proof}[Proof of \cref{thm:eer}]
We apply \cref{lem:truncation} to $\widetilde X$ with $\theta=\widetilde H^{-2}\le\frac1{25}$.
For brevity, we define $\widetilde H_F=H(\widetilde X_F)$ and let $\widetilde u_F,\widetilde m_F$ be the two corresponding quantities formed from the zero-free finitely supported $\widetilde X_F$.
Then by \cref{lem:truncation},
\begin{align*}
    \widetilde m
    \geq (1-\theta)^2\widetilde m_F \frac{\widetilde H_F}{\widetilde H}
    \geq \kappa\widetilde m_F
    \geq \widetilde m_F - \frac{26}{5}\widetilde H^{-2},
\end{align*}
where we use the fact that $\widetilde m_F \leq 2$ and let
\[ \kappa=(1-\theta)^2\paren{1-\frac{3\theta}{\widetilde H}} \geq 1-2\theta-\frac{3\theta}{\widetilde H} \geq 1 - \frac{13}{5}\widetilde H^{-2} \geq 1-\frac{104}{1000}. \]
\eqref{eq:eer-clean} then yields
\begin{align*}
    \widetilde m_F \geq 2-6\widetilde u_F - \vartheta(\widetilde H_F),
\end{align*}
where $\vartheta(\eta)=\frac{18\log\eta+57}{\eta}$ for $\eta\geq4$ (in other words, we use the $\vartheta$ defined in \cref{cor:eer-eff}, rather than that defined in \cref{thm:eer}).
We bound
\begin{align*}
    1+\widetilde u_F
    &= \frac{H(\widetilde X_F+\widetilde X_F')}{\widetilde H_F}
    \leq \frac{(1+\widetilde u)\widetilde H}{(1-\theta)^2\widetilde H_F}
    \leq \frac{1+\widetilde u}{\kappa},
\end{align*}
so, using $\widetilde u\leq1$,
\begin{align*}
    6\widetilde u_F \leq 6\widetilde u + 12\frac{1-\kappa}{\kappa} \leq 6\widetilde u + \frac{156}{5\kappa}\widetilde H^{-2} \leq 6\widetilde u + 35\widetilde H^{-2}.
\end{align*}
And as $\vartheta$ is nonincreasing,
\begin{align*}
    \vartheta(\widetilde H_F) \leq \vartheta(\widetilde H-3\theta) \leq \vartheta(\widetilde H)\frac{\widetilde H}{\widetilde H-3\theta} = \vartheta(\widetilde H) + \frac{3\vartheta(\widetilde H)}{\widetilde H^3-3} \leq \vartheta(\widetilde H) + \frac{3\vartheta(5)}{\frac{122}{25}\widetilde H^2} \leq \vartheta(\widetilde H) + \frac{61}{5}\widetilde H^{-2}.
\end{align*}
Combining these bounds yields
\[ \widetilde m \geq 2 - 6\widetilde u - \vartheta(\widetilde H) - \paren{\frac{26}{5}+35+\frac{61}{5}}\widetilde H^{-2} \geq 2-6\widetilde u - \vartheta(\widetilde H) - 11\widetilde H^{-1} \geq 2-6\widetilde u - \frac{18\log\widetilde H + 57+11}{\widetilde H}, \]
which completes the proof.
\end{proof}
The same truncation extends \cref{prop:baseline} as well, yielding \cref{thm:baseline}.
\begin{proof}[Proof of \cref{thm:baseline}]
We apply \cref{lem:truncation} to $\widetilde X$ with $\theta=\widetilde H^{-2}\le\frac1{25}$.
Then, using the same estimates as in the proof of \cref{thm:eer}, by \cref{prop:baseline},
\begin{align*}
    \widetilde m
    \geq \widetilde m_F - \frac{26}{5}\widetilde H^{-2}
    \geq \frac54-\frac34\widetilde u_F - \vartheta(\widetilde H_F) - \frac{26}{5}\widetilde H^{-2},
\end{align*}
where $\vartheta$ is defined as in \cref{prop:baseline}, not \cref{thm:baseline}.
Analogous to before, we bound
\[ \frac34\widetilde u_F \leq \frac34\widetilde u + \frac{3}{2}\cdot\frac{13}{5\kappa}\widetilde H^{-2} \leq \frac34\widetilde u + \frac{22}{5}\widetilde H^{-2}, \]
and
\[ \vartheta(\widetilde H_F) \leq \vartheta(\widetilde H) + \frac{3\vartheta(5)}{\frac{122}{25}\widetilde H^2} \leq \vartheta(\widetilde H) + \paren{\frac{3}{5}+\frac{7}{100}K_R}\widetilde H^{-2}. \]
Thus,
\begin{align*}
    \widetilde m
    &\geq \frac54-\frac34\widetilde u - \vartheta(\widetilde H) - \paren{\frac{51}{5}+\frac{7}{100}K_R}\widetilde H^{-2}
    \\ &\geq \frac54-\frac34\widetilde u - \frac{\frac34\log(1+\widetilde H)+\paren{\frac12+\frac{7}{500}}K_R+\paren{\frac34(\log e + 2)+\frac14+\frac{51}{25}}}{\widetilde H}
    \\ &\geq \frac54-\frac34\widetilde u - \frac{\frac34\log(1+\widetilde H)+K_R+5}{\widetilde H}. \qedhere
\end{align*}
\end{proof}
\begin{remark}\label{rmk:characteristic}
The ground field, in our case $\R$, matters twice.
\Cref{sec:incidence} needs characteristic $0$, so that hypothesis (ii) of \cref{thm:incidence} is vacuous; over $\F_p$ one would have to carry $\abs{G_1}\abs{G_2}\abs{G_3}\ll p^{2}$ through the dyadic uniformized layering.
More seriously, \cref{sec:solymosi,sec:eer} use the \emph{order structure} of $\R$ essentially, through the ordering of the slopes and the strict convex combination in the proof of \cref{thm:asym-solymosi}, so the Solymosi half, i.e., the $F_2$ argument, has no clear analog outside of $\R$.
However, as discussed in \cref{rmk:1_9}, without the Solymosi half we can still achieve $\delta=\frac19$, and so this result holds for arbitrary fields of characteristic 0.
\end{remark}

\section{Open problems}\label{sec:open}
The main conjecture we have is \cref{conjecture:entropic_solymosi}, which as discussed would imply \cref{conjecture:main} for $\delta=\frac13$, and this would be tight by the example of the author, Gavalakis, and Kontoyiannis \cite[Example 3]{LGK} showing $\delta\leq\frac13$; the example is naturally also asymptotically tight for \cref{conjecture:entropic_solymosi}.
Recall their example has $X=0$ with probability $\frac13$ and otherwise is uniformly distributed in $\set{1,2,\dots,n}$.

We can generalize their example to arbitrary $p=\P(X=0)$, though for simplicity will shift the arithmetic progression from $\set{1,2,\dots,n}$ to $\set{N+1,\dots,N+n}$ for $N=2n^2$ to better separate the products.
While $p=\frac13$ is optimal for \cref{conjecture:main} in that $p\neq\frac13$ will provide an upper bound on $\delta$ larger than $\frac13$, we note that any $p\in[0,1)$ yields a tight example for \cref{conjecture:entropic_solymosi}, as follows.
Note this construction is \cref{ex:16}, and so we have already computed the entropies to be
\[ H(X)=(1-p)\log n + O(1), \quad H(X+X')=(1-p^2)\log n + O(1), \quad H(XX') = 2(1-p)^2\log n + O(1). \]
Then
\[ 2H(X+X')+H(XX') \sim (4-4p)\log n \sim 4H(X). \]

More generally, if we have any zero-free law $\widetilde X$ and define $X$ by adding an atom at 0 with probability $p$, i.e., $\widetilde X\sim X\mid X\neq 0$ and $\P(X=0)=p$, then by our computations in \cref{lem:transfer}, where our lower bound on $H(X+X')$ is tight up to an additive constant of 2, as we conditioned on two indicator variables, we have $H=(1-p)\widetilde H + O(1)$,
\[ H(X+X') = 2p(1-p)\widetilde H + (1-p)^2 H(\widetilde X + \widetilde X') + O(1), \quad\text{and}\quad H(XX') = (1-p)^2H(\widetilde X\widetilde X') + O(1). \]
Thus, if we define
\[ \widetilde\Lambda = 2H(\widetilde X+\widetilde X')+H(\widetilde X\widetilde X')-4\widetilde H, \]
we analogously have
\begin{align*}
    \Lambda
    &= 2H(X+X')+H(XX')-4H
    \\ &= 2(1-p)^2H(\widetilde X+\widetilde X')+(1-p)^2H(\widetilde X\widetilde X')-4(1-p)^2\widetilde H + O(1)
    \\ &= (1-p)^2\widetilde\Lambda + O(1).
\end{align*}
In other words, adding an atom at 0 essentially contracts the difference between the two sides of \cref{conjecture:entropic_solymosi} towards 0, and so in some sense this conjecture is ``zero-invariant.''
\begin{remark}\label{rmk:conj-es-zf}
    Concretely, this shows that the general case of \cref{conjecture:entropic_solymosi} follows from the zero-free case.
    If for all $\eps>0$ we knew $\widetilde\Lambda \geq -\eps\widetilde H-O(1)$ for all zero-free laws $\widetilde X$ with entropy $\widetilde H=H(\widetilde X)$, which is equivalent to \cref{conjecture:entropic_solymosi} for zero-free laws via diagonalization, then for any $X$ with $\P(X=0)=p$ and $\widetilde X\sim X\mid X\neq 0$,
    \[ \Lambda = (1-p)^2\widetilde\Lambda - O(1) \geq -\eps(1-p)^2\widetilde H - O(1) = -\eps(1-p) H - O(1) \geq -\eps H - O(1), \]
    which by diagonalization implies \cref{conjecture:entropic_solymosi} for all discrete random variables $X$, not necessarily zero-free.
\end{remark}
One avenue for further research would be to improve the coefficient $6$ of \cref{thm:eer}, or prove that it is tight.
We also note that the improved function need not be linear in $\widetilde u$, though a nonlinear shape may be unnatural to prove.
\cref{prop:master} already establishes how such an improvement would translate to an improvement on $\delta$.
Recall the decomposition of the coefficient $6=2\times(1+2)$ from \cref{rmk:shape}.
Within this proof structure, where of course it is possible a different approach will yield a better coefficient, we believe the outer factor of 2 from the Solymosi argument in \cref{sec:solymosi} is structural, i.e., cannot be improved.
The inner term of 2 comes from the application of the entropic Ruzsa triangle inequality, \eqref{eq:difference-from-sum}, i.e.,
\[ H(X-X') - H(X) \leq 2(H(X+X')-H(X)). \]
While the constant 2 in the combinatorial analog,
\[ \frac{\abs{A-A}}{\abs{A}}\leq\paren{\frac{\abs{A+A}}{\abs{A}}}^2, \]
is known to be tight by \cite{HennecartRobertYudin1999}, it is not known whether this constant is tight in the entropic case.
An improvement here would immediately improve the coefficient of 6 and thus improve $\delta$ beyond $\frac17$, not to mention the improvement being of immense independent interest.

Recall the \emph{varentropy} of a discrete random variable $X$ with probability mass function $p$ is $\Var(-\log p(X))$, where the entropy is simply $\E[-\log p(X)]$.
A second avenue for further research would be to control the varentropy of sums, to essentially allow for a second moment method to improve our control in \cref{sec:eer}, e.g., replacing our need for a $\beta$-adic decomposition in \cref{lem:levelsets}.

\section*{Acknowledgments}
Rupert Li was partially supported by a Hertz Fellowship and a PD Soros Fellowship.
Claude was used as an auxiliary tool while preparing this manuscript.
All statements, arguments, and proofs were verified and either written or revised by the author, who takes full responsibility.

\bibliographystyle{amsinit}
\bibliography{ref}

\end{document}